\documentclass[reqno]{amsart}%
\usepackage{amsfonts}
\usepackage{verbatim}
\usepackage{xcolor}
\usepackage{amsmath}
\usepackage{amssymb}
\usepackage{graphicx}
\usepackage{accents}%
\providecommand{\U}[1]{\protect\rule{.1in}{.1in}}
\newtheorem{theorem}{Theorem}[section]

\newtheorem{definition}[theorem]{Definition}
\newtheorem{proposition}[theorem]{Proposition}
\newtheorem{remark}[theorem]{Remark}
\newtheorem{lemma}[theorem]{Lemma}

\numberwithin{equation}{section}
\begin{document}
\title[KdV equation]{A continuous analog of the binary Darboux transformation for the Korteweg-de
Vries equation\\ }
\dedicatory{This paper is dedicated to Vladimir Marchenko on the occasion of his hundredth
birthday. The author has profoundly been influenced by his groundbreaking
research on inverse problems and completely integrable systems.}\author{Alexei Rybkin}
\address{Department of Mathematics and Statistics, University of Alaska Fairbanks, PO
Box 756660, Fairbanks, AK 99775, USA}
\email{arybkin@alaska.edu}
\thanks{The author is supported in part by the NSF grant DMS-2009980. The author would
like to thank the Isaac Newton Institute for Mathematical Sciences for support
and hospitality during the programme Dispersive Hydrodynamics when work on
this paper was undertaken (EPSRC Grant Number EP/R014604/1).}
\date{March, 2023}
\subjclass{34L25, 37K15, 47B35}
\keywords{KdV equation, Darboux transformation, Riemann-Hilbert problem.}

\begin{abstract}
In the KdV context we put forward a continuous version of the binary Darboux
transformation (aka the double commutation method). Our approach is based on
the Riemann-Hilbert problem and yields a new explicit formula for perturbation
of the negative spectrum of a wide class of step-type potentials without
changing the rest of the scattering data. This extends the previously known
formulas for inserting/removing finitely many bound states to arbitrary sets
of negative spectrum of arbitrary nature. In the KdV context our method offers
same benefits as the classical binary Darboux transformation does.

\end{abstract}
\maketitle

\section{Introduction}

As the title suggests, we are concerned with the \emph{binary Darboux
transformation} in the context of the \emph{Korteweg-de Vries equation} (KdV).
The literature on the Darboux transformation goes back to the nineteenth
century and is immensely extensive and diverse. We only review some of what is
directly related to our paper and where the interested reader may find further
references. The very term appears to be introduced in 1985 by
Babich-Matveev-Salle \cite{Babichetal85} in the context of the Toda lattice
and then extended to the Kadomtsev-Petviashvili equation (see the influential
1991 book \cite{MatveevSalle91} by Matveev and Salle). The name owes to the
fact that the \emph{single Darboux transformation} (also know as Crum,
elementary, or standard) is applied twice: to the associated AKNS system and
its conjugate. We also refer the interested reader to Ling et al
\cite{LingetalNON2015} and the extensive literature cited therein and to
Cieslinski \cite{Ci2009} where the binary Darboux transformation is revisited
from a different point of view. Note that the binary Darboux transformation
was originally introduced to generate explicit solutions to integrable systems
by algebraic means (c.f. \cite{GuetalBook05,MatveevSalle91}) and the
\emph{inverse scattering transform} (IST) was not directly used (while our
approach is based on the IST).

It is interesting to note that what in the Darboux transformation community is
referred to as the binary Darboux transformation is, in fact, also known in
spectral theory of Sturm-Liouville operators as the \emph{double}
\emph{commutation method }introduced in 1951 by Gelfand and Levitan in their
seminal paper \cite{GelfandLevitan55} in the context of their ground breaking
study of the inverse spectral problem for Sturm-Liouville operators. The name
seems however to be cast by Deift \cite{DeiftDuke78} in 1978 as the method
rests on applying twice a commutation formula from operator theory. Note that
Gelfand and Levitan did not use commutation arguments but relied on
\emph{transformation operator} techniques (see also the book \cite[Section
6.6]{Levitan87} by Levitan). The full treatment of the double commutation
method is given by Gesztesy et al \cite{GesztesyetalTAMS91}%
-\cite{GesztTeschl96} in the 1990s (see also the extensive literature cited
therein). The double commutation method was introduced to study the effect of
inserting/removing eigenvalues in spectral gaps on spectral properties of the
underlying 1D Schr\"{o}dinger operators while the binary Darboux
transformation, as we have mentioned, has been primarily a tool to produce
explicit solutions. This is likely a reason why we could not find the
literature where the two would be explicitly linked\footnote{E.g. the book
\cite{GuetalBook05} pays much of attention to binary Darboux transformations
but double commutation is not mentioned. The recent \cite{Sahknovich17}
briefly mentiones \cite{GuetalBook05} and \cite{GesztTeschl96} but without
discussing connections.}.

The main feature of the Darboux transformation (both, single and binary) is
that it allows us to add or remove finitely many eigenvalues of the underlying
system without altering the rest of the spectrum, which offers a powerful tool
to study completely integrable systems. We are concerned with altering certain
types of continuous spectrum too. More specifically, in the context of the KdV
equation we introduce a broad class of the initial profiles, referred below to
as \emph{step-type} (see Definition \ref{step-type}), that admits an extension
of the binary Darboux transformation allowing us to perturb (in particular,
add or remove) the negative spectrum in nearly unrestricted way without
affecting the rest of the scattering data. This class includes, as a very
particular case, initial profiles approaching different constant values at
$\pm\infty$ that recently drew renewed interest (see e.g. the recent Egorova
et al \cite{Egorova2022} and Girotti et al \cite{Grava21} and the extensive
literature cited therein). We start out from the Riemann-Hilbert version of
the binary Darboux transformation put forward in our recent \cite{Ryb21} and
show that our construction is not really restricted to isolated negative
eigenvalues and can be readily extended to the negative spectrum of arbitrary
nature that step-type potentials can produce. As is well-known, Darboux
transformations are particularly convenient in studying soliton propagation
over various backgrounds and it is reasonable to expect that our approach can
be used to the same effect for whole intervals of continuous negative
spectrum. This becomes particularly relevant in the light of the recent spike
of interest to soliton gases (see section \ref{main results}).

The paper is organized as follows. In Section \ref{Notation} we outline our
notational agreements. In Section \ref{background info} we go over some basics
necessary to fix our notation and terminology as well as some of our previous
results needed in what follows. Section \ref{main results} is devoted to the
statement and discussions of our main result, Theorem \ref{MainThm}. In
Section \ref{MRHP} we rewrite the classical formulation of the Riemann-Hilbert
problem approach to the KdV equation in the form convenient for our
generalizations. In Section \ref{QARHP} we state our Riemann-Hilbert problem
with a jump matrix which entries are distributions. In Section \ref{cont DT}
we state and prove a Riemann-Hilbert version of the continuous binary Darboux
transformation. Theorem \ref{Darboux for RHP} proven therein is essentially
equivalent to Theorem \ref{MainThm} but we hope it may present an independent
interest, especially to the reader who prefers the Riemann-Hilbert problem
framework. Section \ref{discrete oper} is auxiliary and devoted to
discretization of our main integral operator which our limiting arguments are
based on. Section \ref{log-det} is another auxiliary section where a
convenient formula is derived for the proof of the main theorem. In Section
\ref{proof of main theorem} we finally prove Theorem \ref{MainThm} which
essentially amounts to combining the ingredients prepared in the previous
sections. In Section \ref{sect on reflectionless} we consider explicit
examples. The first one is a new derivation of the well-known formula for pure
soliton solutions and the second one is an explicit construction of
reflectionless step-type potential. Appendix is devoted to some more auxiliary statements.

\section{Notation\label{Notation}}

Through the paper, we make the following notational agreement. The bar denotes
the complex conjugate. Prime $^{\prime}$ means the derivative (perhaps
generalized) in the main variable (typically spatial or in spectral
parameter). The temporal variable $t$ appears only as a parameter and is
frequently suppressed. $W\{f,g\}=fg^{\prime}-f^{\prime}g$ is the Wronskian
with obvious interpretation if one of $f$ is a vector.

Given a non-negative finite Borel measure $\mu$ on the real line,
$L^{2}\left(  \mathrm{d}\mu\right)  $ is the real Hilbert space with the inner
product $\left\langle f,g\right\rangle =\int f\left(  x\right)  g\left(
x\right)  $\textrm{d}$\mu\left(  x\right)  $, where the integral is taken over
the support $\operatorname*{Supp}\mu$. In particular, d$\mu=$d$x$, the
Lebesgue measure on the whole real line, then we conveniently abbreviate
$L^{2}\left(  \mathrm{d}x\right)  =L^{2}$. We apply the same agreement to
other Lebesgue spaces $L^{p}$. Thus, we conveniently write%
\[
\int\cdot\ \mathrm{d}\mu=\int_{\operatorname*{Supp}\mu}\cdot\ \mathrm{d}%
\mu,\int_{-\infty}^{\infty}=\int.\ \
\]
If a function $f\left(  x\right)  $ is defined on a set larger than
$\operatorname*{Supp}\mu$ we write $f\in L^{p}\left(  \mathrm{d}\mu\right)  $
if its restriction to $\operatorname*{Supp}\mu$ is in $L^{p}\left(
\mathrm{d}\mu\right)  $. We adopt the following notation common in scattering
theory:%
\[
L^{1}\left(  \left(  1+\left\vert x\right\vert \right)  ^{N}\mathrm{d}%
x\right)  =L_{N}^{1},\ \ \ N\geq1.
\]
$\chi_{S}$ is the characteristic function of a set $S$. In particular,
$\chi_{\pm}=\chi_{\mathbb{R}_{\pm}}$ for the Heaviside functions. We write%
\[
f\in L^{p}\left(  \pm\infty\right)  \text{ if }f\chi_{\left(  a,\pm
\infty\right)  }\in L^{p}\text{ for any finite }a\text{.}%
\]
The same convention is used for $L_{N}^{p}\left(  \pm\infty\right)  $. If $S$
is not a subset of $\operatorname*{Supp}\mu$ we write $\mu\left(  S\right)
=\mu\left(  S\cap\operatorname*{Supp}\mu\right)  $.

We say that a sequence $f_{n}\left(  x\right)  $ converges in the weak* sense
to $f\left(  x\right)  $ on $S\subseteq\mathbb{R}$ if%
\[
\lim_{n\rightarrow\infty}\int_{S}\left(  f\left(  x\right)  -f_{n}\left(
x\right)  \right)  \varphi\left(  x\right)  \mathrm{d}x=0
\]
for any $\varphi\in C_{c}^{\infty}\left(  S\right)  $ (i.e. an infinitely
differentiable compactly supported function on $S$). We will refer to
$\mathrm{d}\mu\left(  x\right)  /\mathrm{d}x$ as generalized
density/derivative or distributional derivative. It is defined by%
\[
\int\varphi\left(  x\right)  \frac{\mathrm{d}\mu\left(  x\right)  }%
{\mathrm{d}x}\mathrm{d}x=\int\varphi\left(  x\right)  \mathrm{d}\mu\left(
x\right)
\]
for any $\varphi\in C_{c}^{\infty}$. In particular, $\delta_{a}\left(
x\right)  =\mathrm{d}\chi_{\left(  a,+\infty\right)  }\left(  x\right)
/\mathrm{d}x$ is the Dirac delta-function supported at $a$. We call a measure
$\mu$ symmetric if $\mathrm{d}\mu\left(  x\right)  =\mathrm{d}\mu\left(
\left\vert x\right\vert \right)  $.

Square matrices are denoted by boldface capital letters, except for the
identity matrix $I$. We use boldface lowercase letters for vectorial (row or
column) quantities. Except for $\mathbb{R}$, $\mathbb{C}$, blackboard bold
letters denote operators. Less common notations and conventions will be
introduced later.

\section{Brief Review of the IST Method\label{background info}}

We are concerned with the Cauchy problem for the KdV equation%
\begin{equation}
\partial_{t}u-6u^{\prime}+u^{\prime\prime\prime}=0,\ \ \ x,t\in\mathbb{R},
\label{KdV}%
\end{equation}%
\begin{equation}
u\left(  x,0\right)  =q\left(  x\right)  , \label{KdVID}%
\end{equation}
the first nonlinear evolution PDE solved in 1967 by Gardner, Greene, Kruskal,
and Miura \cite{GGKM67} by the\emph{\ }inverse scattering transform (IST). For
the reader's convenience and to fix our notation we review the necessary
material following \cite{AC91,MarchBook2011,NPZ}. The IST method consists, as
the standard Fourier transform method, of\ three steps:

\textbf{Step 1. }(direct transform)%
\[
q\left(  x\right)  \longrightarrow S_{q},
\]
where $S_{q}$ is a new set of variables which turns (\ref{KdV}) into a simple
first order linear ODE for $S_{q}(t)$ with the initial condition
$S_{q}(0)=S_{q}$.

\textbf{Step 2. }(time evolution)%
\[
S_{q}\longrightarrow S_{q}\left(  t\right)  .
\]

\textbf{Step 3. }(inverse transform)%
\[
S_{q}\left(  t\right)  \longrightarrow q(x,t).
\]

\subsection{Classical short-range IST}

Step 1. Suppose that the initial condition $q$ for (\ref{KdVID}) is real and
rapidly decaying. This means that the solution $q\left(  x,t\right)  $ to
(\ref{KdV})-(\ref{KdVID}) is subject to the decay condition%
\begin{equation}
\int_{-\infty}^{\infty}\left(  1+\left\vert x\right\vert \right)  \left\vert
q\left(  x,t\right)  \right\vert \mathrm{d}x<\infty,\ \ \ t\geq0,\ \ \text{
(short-range).}\label{decay cond}%
\end{equation}
Associate with $q\left(  x\right)  $ the full line \emph{Schr\"{o}dinger
operator} $\mathbb{L}_{q}=-\partial_{x}^{2}+q\left(  x\right)  $. (For
simplicity we retain the same notation for the differential operation.) Its
spectrum $\operatorname*{Spec}\mathbb{L}_{q}$ consists of a two fold
absolutely continuous part filling $[0,\infty)$ and finitely many negative
simple eigenvalues (bound states) $\left(  -\kappa_{n}^{2}\right)  _{n=1}^{N}%
$. An important feature of the short-range $q$ is that it supports two
(right/left)\emph{ Jost solutions} $\psi_{\pm}(x,k)$. I.e. solutions of the
\emph{Schr\"{o}dinger equation}%
\begin{equation}
\mathbb{L}_{q}u=-u^{\prime\prime}+q\left(  x\right)  u=k^{2}%
u,\label{schrodinger eq}%
\end{equation}
having plane wave asymptotics at infinity: $\psi_{\pm}(x,k)\sim$%
\textrm{e}$^{\pm\mathrm{i}kx},x\rightarrow\pm\infty$. Note that $\psi_{\pm
}(x,k)$ are determined by $q\left(  x\right)  $ on $\left(  \pm\infty
,x\right)  $ respectively. These solutions are analytic for $\operatorname{Im}%
k>0$, continuous down to the real line where $\psi_{\pm}(x,-k)=\overline
{\psi_{\pm}(x,k)}$. The pair $\{\psi_{+},\overline{\psi_{+}}\}$ forms a
fundamental set for (\ref{schrodinger eq}) and hence%
\begin{equation}
T(k)\psi_{-}(x,k)=\overline{\psi_{+}(x,k)}+R(k)\psi_{+}(x,k),\;k\in
\mathbb{R},\label{basic scatt identity}%
\end{equation}
with some $T$ and$\ R$ called the transmission and (right)\emph{\ }reflection
coefficients\emph{\ }respectively. While totally elementary,
(\ref{basic scatt identity}), called the \emph{right basic scattering
identity,} serves as a foundation for inverse scattering theory. It
immediately follows from (\ref{basic scatt identity}) that%
\begin{equation}
T\left(  k\right)  =\frac{2\mathrm{i}k}{W\left\{  \psi_{-}\left(  x,k\right)
,\psi_{+}\left(  x,k\right)  \right\}  },\label{T}%
\end{equation}
which means that $T\left(  k\right)  $ can be analytically continued into the
upper half plane with simple poles at $\mathrm{i}\kappa_{n}$ where $\psi
_{-}\left(  x,k\right)  $, $\psi_{+}\left(  x,k\right)  $ are linearly
dependent. Moreover, $\psi_{-}(x,$\textrm{$i$}$\kappa_{n})$, $\psi_{+}%
(x,$\textrm{$i$}$\kappa_{n})$ are real and decay exponentially at both
$\pm\infty$ and hence $-\kappa_{n}^{2}$ is a (negative) bound state
(eigenvalue) of $\mathbb{L}_{q}$. The number%
\begin{equation}
c_{n}=\left(  \int\psi_{+}(x,\mathrm{i}\kappa_{n})^{2}\mathrm{d}x\right)
^{-1/2}\label{cn}%
\end{equation}
is called the \emph{norming constant} of the bound state $-\kappa_{n}^{2}$.
The reflection coefficient $R\left(  k\right)  $ is continuous but need not
analytically extend outside the real line. It obeys $R\left(  -k\right)
=\overline{R\left(  k\right)  }$, $\left\vert R\left(  k\right)  \right\vert
\leq1$, $\left\vert R\left(  k\right)  \right\vert =1$ only at $k=0$.

The main feature of the short-range case is that $R\left(  k\right)  $
(continuous component) and $\left\{  (\kappa_{n},c_{n}),1\leq n\leq N\right\}
$ (discrete component) determine the potential $q$ uniquely. It will be
convenient for our purposes to write the discrete component as the (discrete)
measure%
\begin{equation}
\mathrm{d}\rho\left(  k\right)  =%
{\displaystyle\sum\limits_{n=1}^{N}}
c_{n}^{2}\delta_{\kappa_{n}}\left(  k\right)  \mathrm{d}k,
\label{rou for short-range}%
\end{equation}
which alone carries over all necessary information about the discrete
spectrum. We can now introduce the (right) \emph{scattering data}%
\[
S_{q}:=\{R,\mathrm{d}\rho\}.
\]
We emphasize that $S_{q}$ determines $q$ uniquely (i.e. $S_{q}$ is indeed
data) in general only for short-range $q$'s. (See e.g. our recent
\cite{RybPosBS} for explicit counterexamples.)

Step 2. The main reason why the IST works is that the (necessarily unique)
solution $q\left(  x,t\right)  $ to the problem (\ref{KdV})-(\ref{decay cond})
gives rise to the "time evolved" Schr\"{o}dinger operator $\mathbb{L}%
_{q\left(  \cdot,t\right)  }$ for which%
\[
R\left(  k,t\right)  =R\left(  k\right)  \mathrm{e}^{8\mathrm{i}k^{3}t}%
,\kappa_{n}\left(  t\right)  =\kappa_{n},c_{n}\left(  t\right)  =\mathrm{e}%
^{4\kappa_{n}^{3}t}c_{n}.
\]
Since $q\left(  x,t\right)  $ is short-range for $t>0$%
\[
S_{q}\left(  t\right)  =\left\{  R\left(  k\right)  \mathrm{e}^{8\mathrm{i}%
k^{3}t},\mathrm{e}^{8k^{3}t}\mathrm{d}\rho\left(  k\right)  :k\geq0\right\}
\]
is ("time evolved") scattering data for $q\left(  x,t\right)  $.

Step 3. Solve the inverse scattering problem for $S_{q}\left(  t\right)  $ by
any applicable method.

Note that this scheme runs for the left scattering data equally well. Another
remarkable feature of the classical short-range IST is that potentials subject
to (\ref{decay cond}) can be characterized in term of the scattering data
(known as \emph{Marchenko's characterization} \cite{MarchBook2011}).

\subsection{One sided IST}

The IST was extended in \cite{Hruslov76} and \cite{Venak86} to step-like
initial data $q$, i.e. $q$'s approaching two different values as
$x\rightarrow\pm\infty$ rapidly enough (aka an initial hydraulic jump or bore
wave) or and then to $q$'s approaching periodic function on one end and a
constant on another in \cite{KK94}. The main difference from the short-range
case is that the measure $\rho$ in the scattering data gains an absolutely
continuous component and the right and left ISTs are different. In fact,
\cite{Hruslov76} uses the left IST and \cite{Venak86} does the right IST to
prove that the absolutely continuous component of $\rho$ gives rise to an
infinite sequence of asymptotic solitons twice as high as the initial
hydraulic jump. The latter is the main feature of such initial profiles.
Another interesting feature of step-like initial data is appearance of
rarefaction waves studied recently in \cite{Egorova2022} and \cite{Grava21}.

In \cite{GryRybBLMS20} we extended IST\ to initial data that approach zero at
$+\infty$ fast enough but are essentially arbitrary at $-\infty$. As opposed
to the step-like data studied in \cite{Hruslov76,Venak86} such potentials
support in general only the right inverse scattering. The approach is based on
the notion of a Weyl solution. Recall that a real-valued locally integrable
potential $q$ is said to be Weyl limit point at $\pm\infty$ if the equation
$\mathbb{L}_{q}u=k^{2}u$ has a unique (up to a multiplicative constant)
solution (called Weyl) $\Psi_{\pm}(\cdot,k^{2})$ that is square integrable at
$\pm\infty$ for each $k^{2}\in\mathbb{C}^{+}$. Note that as apposed to Jost
solutions, Weyl solutions exist under much more general conditions on $q$'s
and no decay of any kind is required. As is well-known, if $q\left(  x\right)
$ is limit point at both $\pm\infty$ then $\mathbb{L}_{q}$ is selfadjoint on
$L^{2}$.

It is convenient to give a special name to potentials that we shall deal with
through the rest of the paper.

\begin{definition}
[Step-type potentials]\label{step-type pots} We call a (real) potential
$q\left(  x\right)  $ \emph{step-type }if

\begin{enumerate}
\item $q\left(  x\right)  \in L_{1}^{1}\left(  +\infty\right)  $ (short-range
at $+\infty$);

\item $\operatorname*{Spec}\mathbb{L}_{q}>-\infty$ (essential boundedness below).
\end{enumerate}
\end{definition}

The (right-sided) scattering theorem for such potentials is studied in our
\cite[Section 7]{GruRybSIMA15} where the interested reader can finds the
details. Choose the Weyl solution $\varphi\left(  x,k\right)  $ (note our
variable $k$ not $k^{2}$) at $-\infty$ to satisfy%
\begin{equation}
\varphi\left(  x,k\right)  =\overline{\psi(x,k)}+R(k)\psi(x,k),\text{ a.e.
}\operatorname{Im}k=0, \label{eq8.1}%
\end{equation}
with some coefficient $R\left(  k\right)  $ (c.f. (\ref{basic scatt identity}%
)), which can be called the (right) reflection coefficient. Note that in the
short-range case
\begin{equation}
\varphi\left(  x,k\right)  =T\left(  k\right)  \psi_{-}(x,k).
\label{fi for short-range}%
\end{equation}

\begin{proposition}
[On reflection coefficient]\label{props of R} The (right) reflection
coefficient $R$ of a step-type potential $q\left(  x\right)  $ is well
defined, symmetric $R\left(  -k\right)  =\overline{R\left(  k\right)  }$, and
$\left\vert R\left(  k\right)  \right\vert \leq1$ a.e. Moreover, if
$\Delta\subseteq\operatorname*{Spec}\left(  \mathbb{L}_{q}\right)  $ is the
minimal support of the two fold a.c. spectrum of $\mathbb{L}_{q}$ then
$\left\vert R\left(  k\right)  \right\vert <1$ for a.e. real $k$ such that
$k^{2}\in\Delta$ and $\left\vert R(k)\right\vert =1$ otherwise.
\end{proposition}

This statement describes the positive spectrum of $\mathbb{L}_{q}$ only. The
negative spectrum is described in

\begin{proposition}
[On norming measure]\label{prop on norming measure}If $q\left(  x\right)  $ is
step-type then on the imaginary line $\varphi$ and $\psi$ are related
by\footnote{Here and below, as always, $\operatorname{Im}f\left(
x+\mathrm{i}0\right)  \mathrm{d}x=w^{\ast}-\lim\operatorname{Im}f\left(
x+\mathrm{i}y\right)  \mathrm{d}x,\ y\rightarrow+0.$}%
\begin{equation}
\operatorname{Im}\varphi\left(  x,\mathrm{i}k+0\right)  \mathrm{d}k=\pi
\psi\left(  x,\mathrm{i}k\right)  \mathrm{d}\rho\left(  k\right)
,\ \ \ k\geq0, \label{eq for rou}%
\end{equation}
for some non-negative finite measure $\rho$. Moreover,%
\[
\operatorname*{Supp}\rho=\left\{  k\geq0:-k^{2}\in\operatorname*{Spec}\left(
\mathbb{L}_{q}\right)  \right\}  .
\]

\end{proposition}

Proposition \ref{props of R} is proven in \cite[Proposition 7.10]%
{GruRybSIMA15}. Proposition \ref{prop on norming measure} follows from
Proposition 7.12 of \cite[Section 7]{GruRybSIMA15} (but it shall also become
transparent from considerations below). Note that $\psi\left(  x,k\right)  $
is analytic for $\operatorname{Im}k>0$ and $\varphi\left(  x,t\right)  $ is
analytic for $\operatorname{Im}k>0$ away from $\mathrm{i}\operatorname*{Supp}%
\rho$. If $q$ is short-range then d$\rho$ in (\ref{eq for rou}) coincides with
(\ref{rou for short-range}) (see below). For this reason we call $\rho$ the
\emph{norming measure}. Note that the important in the short-range scattering
representation (\ref{fi for short-range}) defining the transmission
coefficient $T\left(  k\right)  $ is lost for a generic step-type potential.
It can of course be introduced, as frequently done in the literature on
bore-type $q$'s discussed above, by replacing the left Jost solution in
(\ref{fi for short-range}) with another Weyl solution which asymptotic in
$x\rightarrow-\infty$ is know for a.e. real $k$. But the transmission
coefficient introduced this way will depend on such choice. We refer the
interested reader to \cite{GesztezyNowell96} for a general framework on
scattering theory with different spatial asymptotics at $\pm\infty$. For our
purposes we do not need the transmission coefficient \ but it may become
necessary to study asymptotic behavior of KdV solutions as $x\rightarrow
-\infty$.

If $q(x)$ is a pure step function, i.e. $q(x)=-h^{2},\;x<0,\;q(x)=0,\;x\geq0$
then $\operatorname*{Spec}\left(  \mathbb{L}_{q}\right)  =(-h^{2},\infty)$ and
purely absolute continuous, $(-h^{2},0)$ and $\left(  0,\infty\right)  $ being
its simple and two fold components respectively. Moreover%
\begin{equation}
R(k)=-\left(  \frac{h}{\sqrt{k^{2}}+\sqrt{k^{2}+h^{2}}}\right)  ^{2}%
,\ \ \mathrm{d}\rho\left(  k\right)  =\frac{2k}{\pi h^{2}}\sqrt{h^{2}-k^{2}%
}\mathrm{d}k. \label{pure step}%
\end{equation}

Let us now look at what happens to Step 1- Step 3 (the "time evolved"
picture). Using the classical short-range IST as a pattern to follow,
introduce
\[
S_{q}=\left\{  R\left(  k\right)  ,\mathrm{d}\rho\left(  k\right)
:k\geq0\right\}
\]
for Step 1. The problem though is that no analog of Marchenko's
characterization is known to date and, while $S_{q}$ can be formed by solving
right sided scattering problem we cannot claim that there is only one $q$ that
corresponds to $S_{q}$. But we move on to Step 2 and form
\[
S_{q}\left(  t\right)  =\left\{  R\left(  k\right)  \mathrm{e}^{8\mathrm{i}%
k^{3}t},\mathrm{e}^{8k^{3}t}\mathrm{d}\rho\left(  k\right)  :k\geq0\right\}
.
\]
This step is formal as there are no general well-posedness results for the KdV
equation with general step-type initial data. Thus we don't know what
$S_{q}\left(  t\right)  $ actually represents and cannot go over to step 3.
These problems can however be detoured by understanding a KdV solution as a
suitable limit. The following definition is convenient for this. It is also
quite natural from the physical point of view.

\begin{definition}
[Step-type KdV\ solutions]\label{step-type} We call $q\left(  x,t\right)  $ a
(right) \emph{step-type KdV solution} with scattering data $S_{q}=\left\{
R\left(  k\right)  ,\mathrm{d}\rho\left(  k\right)  :k\geq0\right\}  $ if for
$t\geq0$

\begin{enumerate}
\item $q\left(  x,t\right)  $ is step-type (in the sense of Definition
\ref{step-type pots});

\item there is a sequence $S_{n}=\left\{  R_{n},\mathrm{d}\rho_{n}\right\}  $
such that each set $S_{n}$ is the scattering data for some short-range (real)
potential $q_{n}\left(  x\right)  $ and in the weak* sense%
\begin{equation}%
\begin{array}
[c]{ccc}%
S_{n} & \rightarrow & S_{q}\\
q_{n}\left(  x,t\right)  & \rightarrow & q\left(  x,t\right)
\end{array}
\text{, (\text{double convergence}),} \label{double conv}%
\end{equation}
where $q_{n}\left(  x,t\right)  $ is the KdV solution with initial data
$q_{n}\left(  x\right)  $.
\end{enumerate}
\end{definition}

The comments below should clarify the nature of this concept and why it could
be convenient.

1. We emphasize again that no well-posedness for step-type KdV solutions is
available in general. The main issue is that it is not clear what Banach space
such solutions can be included to even speak about well-posedness. For
bore-type initial conditions the best known well-posedness result in given in
the recent \cite{Laurens22} (see also \cite{Ze Li21}). These papers also
suggest that well-posedness for (\ref{KdV})-(\ref{KdVID}) for general
step-type $q$'s may be out of reach.

2. In \cite{Cohen1989} the authors use certain IST constructions to give
examples of nonuniqueness of the Cauchy problem for KdV. One example gives a
nontrivial $C^{\infty}$ solution $q\left(  x,t\right)  $ in a domain
$\{(x,t):0<t<H(x)\}$ for a positive nondecreasing function $H$, such that
$q\left(  x,t\right)  $ vanishes to all orders as $t\rightarrow0$. This
solution decays rapidly as $x\rightarrow+\infty$, but cannot be "well behaved"
as $x$ moves left. Further analysis is required to tell if such $q\left(
x,t\right)  $ may in fact be a step-type potential\footnote{Explicit
IST\ constructions given in \cite{Aktosun06} produce double pole singularity
moving left and thus such solutions fail condition 1 of Definition
\ref{step-type}.} or not but these disturbing examples suggest that condition
2 may indeed be needed to single out a solution that we think of as physically relevant.

3. In a more restrictive form, Definition \ref{step-type pots} appears first
in our \cite{RybNON10} (where it is referred to as natural). The problem
(\ref{KdV})-(\ref{KdVID}) with $q\left(  x\right)  $ subject to%
\begin{equation}
\operatorname*{Spec}\mathbb{L}_{q}>-\infty, \label{cond 1}%
\end{equation}%
\begin{equation}
q\left(  x\right)  \in L_{N}^{1}\left(  +\infty\right)  <\infty,\ \ \ N\geq
5/2, \label{cond 2}%
\end{equation}
is studied in our \cite{GruRybSIMA15,GryRybBLMS20}. It is shown that the
sequence $q_{n}\left(  x,t\right)  $ of KdV solutions corresponding to cut-off
approximations $q_{n}\left(  x\right)  =q\left(  x\right)  \chi_{\left(
-n,\infty\right)  }\left(  x\right)  $ of initial data $q\left(  x\right)  $
converges to a classical solution $q\left(  x,t\right)  $ of (\ref{KdV}) (that
is three times continuously differentiable in $x$ and once in $t$) uniformly
on compacts in $\mathbb{R\times R}_{+}$. (That is, convergence $q_{n}%
(x,t)\rightarrow q\left(  x,t\right)  $ in (\ref{double conv}) is actually
much better than weak*.) The sequence of scattering data $S_{n}$, however,
converges in the weak* sense and no better. Recall that, as is well-known
\cite{Deift79}, even in the short-range case cut-off approximations lead to
the sequence $R_{n}\left(  k\right)  $ that fails to converge
uniformly\footnote{See also \cite{RemlingCMP2015} which results shed some
light on reasons why only a weak convergence is to be expected.} to $R\left(
k\right)  $ at $k=0$. Thus, under conditions (\ref{cond 1}) and (\ref{cond 2})
the double convergence (\ref{double conv}) works as follows: ($n\rightarrow
\infty$)%
\begin{align}
q_{n}(x)  &  =q\left(  x\right)  \chi_{\left(  -n,\infty\right)  }\left(
x\right)  \rightarrow q\left(  x\right)  \text{ a.e.}\label{chain 1}\\
&  \Rightarrow%
\begin{array}
[c]{cccc}%
S_{n} & \rightarrow & S_{q} & \text{star-weakly}\\
q_{n}\left(  x,t\right)  & \rightarrow & q\left(  x,t\right)  & \text{locally
uniformly for }t>0
\end{array}
\text{.}\nonumber
\end{align}
That is, we start out from short-range (cut-off) approximations $q_{n}\left(
x\right)  $ that produce a sequence of scattering data $S_{n}$ required in
condition 2, which weak* limit, in turn, corresponds to $q\left(  x,t\right)
=\lim q_{n}\left(  x,t\right)  $. Note again that even though a solution
$q\left(  x,t\right)  $ is in a way natural (as it comes from cut off
approximations) but we have no ground to believe that there are no "spurious"
solutions. For this reason $S_{q}$ may be referred to as scattering data only
for solutions which we call natural.

4. Since the KdV flow is isospectral (already assumed in the Lax pair
formulation of the KdV equation), condition (\ref{cond 1}) holds for $t>0$.
Condition (\ref{cond 2}) is however not time invariant. Technics of
\cite{GryRybBLMS20} can be readily used to tell how the rate of decay at
$+\infty$ drops under the KdV flow (in fact, to tell how the KdV trades decay
at $+\infty$ for gain in smoothness (work in progress)). \cite{Cohen1989}
suggests that if $q\left(  x\right)  \in L_{N}^{1}\left(  +\infty\right)  $,
then$\ q\left(  x,t\right)  \in L_{N-5/4}^{1}\left(  +\infty\right)
\in\ L_{1}^{1}\left(  +\infty\right)  $ and hence $q\left(  x,t\right)
\in\ L_{1}^{1}\left(  +\infty\right)  $ when $N\geq5/2$.

5. The chain (\ref{chain 1})\ is not always convenient. The one%
\begin{equation}
S_{n}\rightarrow S\text{ \ a.e. }\Rightarrow q_{n}(x,t)\rightarrow q\left(
x,t\right)  \text{ a.e. for }t\geq0\label{chain 2}%
\end{equation}
may, as our \cite{RybPosBS} suggests. Below we have to implement both
(\ref{chain 1}) and (\ref{chain 2}) in (\ref{double conv}). This should
explain the meaning of condition 2.

6. The weak* convergence of $q_{n}(x,t)\rightarrow q\left(  x,t\right)  $ in
(\ref{double conv}) is assumed to be on the save side and may be upgraded. The
real upgrade though would of course be indicating a Banach space where it
takes place. The latter does not appear possible unless we impose strong
assumptions on the behavior of $q\left(  x\right)  $ as $x\rightarrow-\infty$.
E.g. \cite{GryRybBLMS20} if%
\begin{equation}%
{\displaystyle\sum\limits_{m=-\infty}^{\infty}}
\left(  \int_{m}^{m+1}\left\vert q\left(  x\right)  \right\vert dx\right)
^{2}<\infty, \label{cond 1'}%
\end{equation}
then $q_{n}(x,t)\rightarrow q\left(  x,t\right)  $ holds at least in the
Sobolev space $H^{-1}\left(  \mathbb{R}\right)  $. The latter is a direct
consequence of well-posedness of the KdV in $H^{-1}$ recently proven in
\cite{Killip19}. (In fact, any $H^{-\varepsilon}$ with $0<\varepsilon\leq1$
will do.)

7. Condition (\ref{cond 1}) is satisfied if%
\begin{equation}
\operatorname*{Sup}\limits_{\left\vert I\right\vert =1}\int_{I}\max\left(
-q\left(  x\right)  ,0\right)  \mathrm{d}x<\infty. \label{convergence}%
\end{equation}
The latter covers a large class of initial profiles without any assumption on
a pattern of behavior at $-\infty$.

\section{Main theorem\label{main results}}

In this section we present our main statement and offer relevant discussions.

\begin{theorem}
[Pertubration of negative spectrum]\label{MainThm} Let $q\left(  x,t\right)  $
be a step-type KdV solution (in the sense of Definition \ref{step-type}) with
the scattering data
\begin{equation}
S_{q}=\left\{  R\left(  k\right)  ,\mathrm{d}\rho\left(  k\right)
:k\geq0\right\}  , \label{S_q}%
\end{equation}
$\psi\left(  x,t,k\right)  $ its right Jost solution, and
\begin{equation}
K\left(  k/\mathrm{i},s;x,t\right)  :=\int_{x}^{\infty}\psi\left(
z,t;k\right)  \psi\left(  z,t;\mathrm{i}s\right)  \mathrm{d}%
z,\ \ \operatorname{Im}k\geq0,s\geq0. \label{kernal}%
\end{equation}
Then for any finite signed measure $\sigma$ supported on a compact set of
$[0,\infty)$ satisfying the conditions%
\begin{equation}
\int\left\vert \mathrm{d}\sigma\left(  k\right)  \right\vert /k<\infty
,\ \ \ \mathrm{d}\rho+\mathrm{d}\sigma\geq0, \label{conds}%
\end{equation}
the Fredholm integral equation%
\begin{align}
y\left(  \alpha\right)  +\int K\left(  \alpha,s;x,t\right)  y\left(  s\right)
\mathrm{d}\sigma_{t}\left(  s\right)   &  =\psi\left(  x,t;\mathrm{i}%
\alpha\right)  ,\ \alpha\in\operatorname*{Supp}\sigma,\label{Fredholm}\\
\mathrm{d}\sigma_{t}\left(  s\right)   &  :=\mathrm{e}^{8s^{3}t}%
\mathrm{d}\sigma\left(  s\right)  ,\nonumber
\end{align}
has a unique solution $y\left(  s;x,t\right)  $ in $L^{2}\left(
\mathrm{d}\sigma\right)  $ and
\begin{equation}
\psi_{\sigma}\left(  x,t;k\right)  =\psi\left(  x,t;k\right)  -\int K\left(
k/\mathrm{i},s;x,t\right)  y\left(  s,x,t\right)  \mathrm{d}\sigma_{t}\left(
s\right)  , \label{DT for psi}%
\end{equation}
is the right Jost solution corresponding to the potential%
\begin{align}
q_{\sigma}\left(  x,t\right)   &  =q\left(  x,t\right) \label{DT for q}\\
&  +2\left[  \int\psi_{\sigma}\left(  x,t;\mathrm{i}s\right)  \psi\left(
x,t;\mathrm{i}s\right)  \mathrm{d}\sigma_{t}\left(  s\right)  \right]
^{2}+4\int\psi_{\sigma}\left(  x,t;\mathrm{i}s\right)  \psi^{\prime}\left(
x,t;\mathrm{i}s\right)  \mathrm{d}\sigma_{t}\left(  s\right)  ,\nonumber
\end{align}
which is a step-type KdV solution with the scattering data
\[
S_{q_{\sigma}}=\left\{  R\left(  k\right)  ,\mathrm{d}\rho\left(  k\right)
+\mathrm{d}\sigma\left(  k\right)  :k\geq0\right\}  .
\]
Moreover,

\begin{enumerate}
\item For $s\in\operatorname*{Supp}\sigma$%
\begin{equation}
\psi_{\sigma}\left(  x,t,\mathrm{i}s\right)  =y\left(  s,x,t\right)  ;
\label{psi}%
\end{equation}

\item $q_{\sigma}\left(  x,t\right)  $ is as smooth as $q\left(  x,t\right)  $
(i.e. $q\left(  x,t\right)  \in C^{\left(  n\right)  }$ $\Leftrightarrow
q_{\sigma}\left(  x,t\right)  \in C^{\left(  n\right)  }$);

\item If $\ 0\notin\operatorname*{Supp}\sigma$ then $q_{\sigma}\left(
x,t\right)  -q\left(  x,t\right)  $ decays exponentially as $x\rightarrow
+\infty$ for every fixed $t>0$;

\item If $\kappa_{0}$ is a pure point of $\sigma$ with a positive weight then
$-\kappa_{0}^{2}$ is a (negative) bound state of $\mathbb{L}_{q_{\sigma}}$
which is embedded if $\kappa_{0}\in\operatorname*{Supp}\rho$;

\item The binary Darboux transformation is invertible in the following sense%
\[
\left(  \psi_{\sigma}\right)  _{-\sigma}=\psi,\text{ \ \ }\left(  q_{\sigma
}\right)  _{-\sigma}=q.
\]

\end{enumerate}
\end{theorem}

Some comments:

1. Theorem \ref{MainThm} is an extension of Theorem 1 of our recent
\cite{Ryb21} where all considerations are conducted in the short-range
setting. It was observed in \cite{Ryb21} (Remark 3) that only the right Jost
solution $\psi$ explicitly appears and the left Jost solution does not
explicitly appear in any formulas. This led us to conjecture that $q$ may be
quite general at $-\infty$. Theorem \ref{MainThm} answers this conjecture in
the affirmative. The relaxation of the decay assumption at $-\infty$ results
in the appearance of rich negative spectrum described by a measure $\rho$ of
arbitrary nature (in the short-range case it is discrete with finitely many
pure points). To remain in the short-range setting we could consider in
\cite{Ryb21} perturbations $\sigma$ supported only at finitely many points. In
Theorem \ref{MainThm} we no longer have to assume this and moreover $\sigma$
can be continuous. For this reason Theorem \ref{MainThm} can also be named the
\emph{continuous binary Darboux transformation} or the continuous analog of
the double commutation method \cite{GesztTeschl96}. Note that the formula
(\ref{DT for q}) appears to be new even in this case when $\sigma$ is discrete
with finitely many points (c.f. \cite{GesztTeschl96,Ryb21}). Furthermore,
binary Darboux transformations are typically written differently for adding
and removing eigenvalues and thus Theorem \ref{MainThm} combines the two.

2. The main feature of the transformation (\ref{DT for q}) is that it allows
us to modify the negative spectrum in a nearly unrestricted way while leaving
the reflection coefficient unchanged.

3. Since $\psi\left(  \cdot,\mathrm{i}s\right)  $ is real, if d$\sigma\geq0$
then the kernel $K\left(  \alpha,s;x,t\right)  $ is positive definite and
(\ref{Fredholm}) is automatically uniquely solvable. However, if d$\sigma<0$
then this is no longer true in general. A counterexample is easily produced by
$R=0$, $\rho=0$, and d$\sigma\left(  k\right)  =-c^{2}\delta\left(
k-\kappa\right)  $d$k,\kappa>0$. Indeed a simple formal computation gives then%
\[
q_{\sigma}\left(  x,t\right)  =-2\partial_{x}^{2}\log\left(  1-c^{2}%
\mathrm{e}^{8\kappa^{3}t-2\kappa x}\right)  ,
\]
which is a singular KdV solution (has moving real double pole). Thus the
second condition (\ref{conds}) cannot be dropped. In other words, removing an
inexistent eigenvalue produces a singular solution. Note that this way our
approach offers a setting for generating singular solutions on nonzero
backgrounds. It would be interesting to compare our method to Wronskian
considerations commonly used in the this context (see, e.g. the influential
\cite{Ma05}).

4. The first condition (\ref{conds}) guarantees that $0$ is not an eigenvalue
of $q_{\sigma}$. If $0$ were an eigenvalue then $q_{\sigma}$ would not be
short-range at $+\infty$. We believe it can be relaxed to read that $\sigma$
is a Carleson measure (see, e.g. for the definition \cite{Koosis}) but cannot
be completely removed.

5. Equations (\ref{DT for psi}) and (\ref{psi}) give an explicit formula for
reconstruction of the Jost solution $\psi_{\sigma}\left(  x,t;k\right)  $ for
any $\operatorname{Im}k\geq0$ via its values on $\operatorname*{Supp}\sigma$.
Note that there is no such formula for a generic analytic function.

6. Step-type potentials admit embedded discrete negative spectrum. Note that
there are no positive embedded bound states if $q$ is short-range at $+\infty$.

7. The formula (\ref{DT for q}) can also be written as%
\[
q_{\sigma}\left(  x,t\right)  =q\left(  x,t\right)  -2\partial_{x}^{2}\log
\det\left[  I+\mathbb{K}\left(  x,t\right)  \right]  ,
\]
where $\mathbb{K}\left(  x,t\right)  $ is a trace class integral operator
acting by the formula%
\[
\mathbb{K}\left(  x,t\right)  f=\int\left[  \int_{x}^{\infty}\psi\left(
z,t;\mathrm{i}\alpha\right)  \psi\left(  z,t;\mathrm{i}s\right)
\mathrm{d}z\right]  f\left(  s\right)  \mathrm{d}\sigma_{t}\left(  s\right)
.
\]

8. As we have mentioned, it is proven in \cite{Hruslov76,Venak86} that a
short-range perturbation of a pure step function, i.e. $q(x)=-h^{2}%
,\;x<0,\;q(x)=0,\;x\geq0$ (hydraulic jump) gives rise to an infinite sequence
of asymptotic solitons of height $-2h^{2}$ (twice as high as the initial
hydraulic jump). Theorem \ref{MainThm} suggests that this effect is a much
more general phenomenon and the fastest soliton always propagates with the
asymptotic velocity $2h^{2}$ where $h^{2}=-\inf\operatorname*{Spec}%
\mathbb{L}_{q}$. Computing asymptotic phases is the most difficult part. (Work
in progress).

9. As one of our referees drew our attention to, if $q=0$ our (\ref{DT for q})
produces a notion of generalized reflectionless potentials that is reminiscent
of the construction due to Lundina \cite{Lundina85} and Marchenko
\cite{Marchenko91}. More specifically, it is proven in \cite{Marchenko91} that
if the integral equation
\begin{align}
&  \mathrm{e}^{-4\kappa^{3}t+\kappa x}\left\{  a\left(  \kappa\right)
y\left(  \kappa\right)  -\frac{1}{2\kappa}\left[  \int\frac{y\left(  s\right)
-y\left(  \kappa\right)  }{s-\kappa}\mathrm{d}\sigma\left(  s\right)
-1\right]  \right\} \nonumber\\
&  =\mathrm{e}^{4\kappa^{3}t-\kappa x}\left\{  \left[  a\left(  \kappa\right)
-1\right]  y\left(  -\kappa\right)  -\frac{1}{2\kappa}\left[  \int%
\frac{y\left(  s\right)  -y\left(  -\kappa\right)  }{s+\kappa}\mathrm{d}%
\sigma\left(  s\right)  -1\right]  \right\}  \label{March}%
\end{align}
is uniquely solvable for $y\left(  \kappa,x,t\right)  $, then%
\begin{equation}
q\left(  x,t\right)  =-2\partial_{x}\int y\left(  \kappa,x,t\right)
\mathrm{d}\sigma\left(  \kappa\right)  \label{march q}%
\end{equation}
satisfies the KdV equation with data $q\left(  x,0\right)  =q\left(  x\right)
$, $a$ and $\sigma$ being related to $q\left(  x\right)  $. It is not obvious
how (\ref{March}) and (\ref{march q}) are related to (\ref{Fredholm}) and
(\ref{DT1}), respectively, but it would certainly be an interesting question
to ask (especially because an open question related to the measure $\sigma$ is
stated in \cite{Marchenko91}). Note that, as apposed to our (\ref{Fredholm}),
solubility of (\ref{March}) is not on the surface. Finally, we also mention
that our methods are very different from \cite{Marchenko91}, where smoothness
of $q\left(  x\right)  $ is essential, while it is not in our construction.

10. And last but not least, we discuss the relevance of Theorem \ref{MainThm}
to soliton gases. Back in 1971, Zakharov \cite{Zakharov(soliton gas)71}
pioneered a statistical description of multisoliton solutions (\emph{rarefied
soliton gas}) which became a big deal in this millennium after the
introduction of \emph{integrable turbulence} and general framework for random
solutions of integrable PDEs in his influential \cite{Zakharov2009}. This
phenomenon was observed in shallow water wind waves in Currituck Sound, NC
\cite{CostaOsbornePRL14} and was experimentally reproduced in a water tank
\cite{RedoretalPRL19} and optical fibers, drawing even greater interest in a
number of research groups (see e.g.
\cite{Condi22,ZakharovetalPhysD2016,El16,El21,Grava21,NabelikZakharov2019})
with different approaches. \emph{Dense soliton gas and condensate}, particular
important from the physical view point, can be modelled as a closure of pure
soliton solutions (c.f.
\cite{ZakharovetalPhysD2016,El2016,Gesztesy-Duke92,ElKom05,El-et-al2011}). We
mention only \cite{ZakharovetalPhysD2016} where the \emph{Zakharov-Manakov
dressing method} \cite{ZakhMan85} was used to yield \emph{primitive
potentials}, which are \emph{one-gap} but neither periodic nor decaying. Such
solutions are parametrized by \emph{dressing functions }$r_{1},r_{2}$ and
essentially only $r_{2}=0$ has been studied rigorously \cite{Grava21} via RHP
technics. For $r_{2}\neq0$ the only case of $r_{1}=r_{2}$ was just considered
in \cite{NabelikZakharov2019} (\emph{elliptic one-gap potential }if
$r_{1}=r_{2}=1$) but the general case is still out of reach. Note that the
dressing method isn't quite IST and cannot solve a Cauchy problem
\cite{Marchenko91}. While seemingly unrelated, Theorem \ref{MainThm} may put
many KdV soliton gas considerations in the context of the IST for the Cauchy
problem for the KdV equation and provide a rigorous framework to study soliton
gases. In fact, in the soliton gas community they actually study statistical
quantities (\emph{density of states}, \emph{effective velocity},
\emph{collision rate}, etc.) of our left step-type KdV\ solutions from Theorem
\ref{MainThm} with $q\left(  x,t\right)  =0$ (zero background) and specific
a.c. $\mathrm{d}\sigma\geq0$ supported on $\left[  -b^{2},-a^{2}\right]  $
with $a>0$. Inclusion of $q\left(  x,t\right)  \neq0$ (nonzero backgrounds)
and $a=0$ (small solitons) into the picture are good open problems. Another
open problem comes from some numerics suggesting that "injection" of a soliton
into \emph{soliton condensates} may locally in time and space "evaporate" the
latter but this effect is not described mathematically. We are yet to look
into these questions but at least Theorem \ref{MainThm} eases our concern
about rather formal realization of limiting (scaling) arguments quite common
in the physical literature on the subject.

\section{Classical meromorphic vector Riemann-Hilbert problem\label{MRHP}}

In this section we review the standard meromorphic vector Riemann-Hilbert
problem that arises from the classical inverse scattering formalism for the
KdV equation following \cite{GT09}. It will be the starting point in our
search for a suitable formulation of the corresponding Riemann-Hilbert problem
for arbitrary step-type potentials. Through this section we denote%
\[
R(k,t)=R(k)\mathrm{e}^{8\mathrm{i}k^{3}t},c_{n}\left(  t\right)
=c_{n}\mathrm{e}^{4\kappa_{n}^{3}t}.
\]

\textbf{Meromorphic Riemann-Hilbert Problem (MRHP)}: Let $S_{q}=\left\{
R,\left(  \kappa_{n},c_{n}\right)  \right\}  $ be the scattering data of a
short-range potential and let%
\[
\mathbf{J}(k,t):=%
\begin{pmatrix}
1-|R(k,t)|^{2} & -\overline{R(k,t)}\\
R(k,t) & 1
\end{pmatrix}
\ \ \text{(\emph{jump matrix}),}%
\]
and $t$ real parameter. Find a row function $\mathbf{v=}\left(
\begin{array}
[c]{ccc}%
\varphi & , & \psi
\end{array}
\right)  $ meromorphic in $\operatorname{Im}k\neq0$ with simple poles $\left(
\pm\mathrm{i}\kappa_{n}\right)  $, such that:

\begin{enumerate}
\item \label{sym cond}\emph{Symmetry condition}:%
\begin{equation}
\overline{\mathbf{v}(\overline{k})}=\mathbf{v}(-k)=\mathbf{v}(k)%
\begin{pmatrix}
0 & 1\\
1 & 0
\end{pmatrix}
,\ \ \ \operatorname{Im}k\neq0. \label{sym conds}%
\end{equation}

\item \label{jump cond}\emph{Jump condition}: The boundary values
$\mathbf{v}(k\pm$\textrm{i}$0)$ are related by%
\begin{equation}
\mathbf{v}(k+\mathrm{i}0)=\mathbf{v}(k-\mathrm{i}0)\mathbf{J}%
(k,t),\ \ \ \operatorname{Im}k=0. \label{jump conds}%
\end{equation}

\item \label{pole conds}\emph{Pole conditions}:%
\begin{equation}
\operatorname*{Res}_{\pm\mathrm{i}\kappa_{n}}\mathbf{v}=\mathrm{i}c_{n}%
^{2}\left(  t\right)  \left\{
\begin{array}
[c]{c}%
\left.  \mathbf{v}%
\begin{pmatrix}
0 & 0\\
1 & 0
\end{pmatrix}
\right\vert _{\mathrm{i}\kappa_{n}}\\
-\left.  \mathbf{v}%
\begin{pmatrix}
0 & 1\\
0 & 0
\end{pmatrix}
\right\vert _{-\mathrm{i}\kappa_{n}}%
\end{array}
\right.  ,\ \ \ 1\leq n\leq N. \label{pole cond}%
\end{equation}

\item \label{asym cond}\emph{Asymptotic condition}\footnote{It is more common
to set $\mathbf{v}(k)\sim\left(
\begin{array}
[c]{ccc}%
\mathrm{1} & , & \mathrm{1}%
\end{array}
\right)  ,k\rightarrow\infty$, to remove oscillations but then the jump matrix
gains an undesirable dependence on $x.$}: For real $x$%
\begin{equation}
\mathbf{v}(k)\sim\left(
\begin{array}
[c]{ccc}%
\mathrm{e}^{-\mathrm{i}kx} & , & \mathrm{e}^{\mathrm{i}kx}%
\end{array}
\right)  ,\ \ \ k\rightarrow\infty. \label{asym conds}%
\end{equation}

\end{enumerate}

As is well-known,%
\begin{equation}
\mathbf{v}(k)=\left(
\begin{array}
[c]{ccc}%
\varphi\left(  k\right)  & , & \psi\left(  k\right)
\end{array}
\right)  , \label{v}%
\end{equation}
where%
\[
\varphi\left(  k\right)  =T\left(  k\right)  \psi_{-}\left(  x,t;k\right)
,\ \ \ \psi\left(  k\right)  =\psi_{+}\left(  x,t;k\right)  .
\]
solves MRHP and the potential $q\left(  x,t\right)  $ is recovered from its
second component. It is important that this component has no poles.

Note that the pole conditions (\ref{pole cond}) can be written in the scalar
form%
\begin{equation}
\operatorname*{Res}_{\mathrm{i}\kappa_{n}}\varphi\left(  k\right)
=\mathrm{i}c_{n}^{2}\left(  t\right)  \psi\left(  \mathrm{i}\kappa_{n}\right)
,\operatorname*{Res}_{-\mathrm{i}\kappa_{n}}\varphi\left(  k\right)
=-\mathrm{i}c_{n}^{2}\left(  t\right)  \psi\left(  \mathrm{i}\kappa
_{n}\right)  , \label{pole cond vector form}%
\end{equation}
which we now represent in an integral form. Due to symmetry it is enough to
consider only the first equation in (\ref{pole cond vector form}). Let $I$ be
an open interval in $\mathbb{R}_{+}$, then one can easily see that%
\begin{align}
\int_{I}\operatorname{Im}\varphi\left(  \mathrm{i}s-0\right)  \mathrm{d}s  &
=\lim_{\varepsilon\rightarrow+0}\int_{I}\operatorname{Im}\varphi\left(
\mathrm{i}s-\varepsilon\right)  \mathrm{d}s\label{*}\\
&  =\lim_{\varepsilon\rightarrow+0}\int_{I}\frac{\varphi\left(  \mathrm{i}%
s-\varepsilon\right)  -\overline{\varphi\left(  \mathrm{i}s-\varepsilon
\right)  }}{2\mathrm{i}}\mathrm{d}s\nonumber\\
&  =-\lim_{\varepsilon\rightarrow+0}\int_{I}\frac{\varphi\left(
\mathrm{i}s-\varepsilon\right)  -\varphi\left(  \mathrm{i}s+\varepsilon
\right)  }{2}\mathrm{id}s\nonumber\\
&  =\lim_{\varepsilon\rightarrow+0}\frac{1}{2}\int_{C_{\varepsilon}\left(
I\right)  }\varphi\left(  z\right)  \mathrm{d}z=\mathrm{i}\pi\sum_{\kappa
_{n}\subset I}\operatorname*{Res}_{\mathrm{i}\kappa_{n}}\varphi.\nonumber
\end{align}
Here $C_{\varepsilon}\left(  I\right)  $ is a contour in $\mathbb{C}^{+}$
enclosing $I$ and shrinking to $I$ as $\varepsilon\rightarrow0$. It follows
now from this and (\ref{pole cond vector form}) that
\begin{align}
\int_{I}\operatorname{Im}\varphi\left(  \mathrm{i}s-0\right)  \mathrm{d}s  &
=-\sum_{\kappa_{n}\subset I}\pi c_{n}^{2}\psi\left(  \mathrm{i}\kappa
_{n}\right) \label{lhs of pole cond v. 1}\\
&  =-\pi\int_{I}\psi\left(  \mathrm{i}s\right)  \mathrm{d}\rho\left(
s\right)  ,\nonumber
\end{align}
where%
\begin{equation}
\mathrm{d}\rho\left(  s\right)  :=\sum_{n}c_{n}^{2}\left(  t\right)
\delta_{\kappa_{n}}\left(  s\right)  \mathrm{d}s, \label{discrete rou}%
\end{equation}
and $\delta_{\kappa_{n}}$ is the Dirac delta function supported at $\kappa
_{n}$. Since $I$ is arbitrary, it follows from (\ref{*}) that%
\[
\operatorname{Im}\varphi\left(  \mathrm{i}s-0\right)  \mathrm{d}s=-\pi
\psi\left(  \mathrm{i}s\right)  \mathrm{d}\rho\left(  s\right)
\]
and, since $\operatorname{Im}\psi\left(  \mathrm{i}s-0\right)  =0$ ($\psi$ is
analytic away from $\mathbb{R}$), we can rewrite now the pole conditions
(\ref{pole cond}) as%
\begin{align}
\operatorname{Im}\mathbf{v}\left(  \mathrm{i}s-0\right)   &  =\left\{
\begin{array}
[c]{c}%
\mathbf{v}\left(  \mathrm{i}s+0\right)
\begin{pmatrix}
0 & 0\\
-\pi\chi_{+}\left(  s\right)  \mathrm{d}\rho\left(  \left\vert s\right\vert
\right)  /\mathrm{d}s & 0
\end{pmatrix}
,\ \ \ s>0\\
-\mathbf{v}\left(  \mathrm{i}s+0\right)
\begin{pmatrix}
0 & -\pi\chi_{-}\left(  s\right)  \mathrm{d}\rho\left(  \left\vert
s\right\vert \right)  /\mathrm{d}s\\
0 & 0
\end{pmatrix}
,\ \ \ s<0
\end{array}
\right. \label{Im v}\\
&  =\mathbf{v}\left(  \mathrm{i}s+0\right)  \left(
\begin{array}
[c]{cc}%
0 & -2\mathrm{i}\pi\chi_{-}\left(  s\right)  \mathrm{d}\rho\left(  \left\vert
s\right\vert \right)  /\mathrm{d}s\\
-2\mathrm{i}\pi\chi_{+}\left(  s\right)  \mathrm{d}\rho\left(  \left\vert
s\right\vert \right)  /\mathrm{d}s & 0
\end{array}
\right)  ,\nonumber
\end{align}
where the derivatives are understood in the sense of distributions and, as
always, $\chi_{\pm}$ is the characteristic function of $\mathbb{R}_{\pm}$.

Since $2\mathrm{i}\operatorname{Im}\mathbf{v}\left(  \mathrm{i}s-0\right)
=\mathbf{v}\left(  \mathrm{i}s-0\right)  -\overline{\mathbf{v}\left(
\mathrm{i}s-0\right)  }=\mathbf{v}\left(  \mathrm{i}s-0\right)  -\mathbf{v}%
\left(  \mathrm{i}s+0\right)  $, (\ref{Im v}) yields%
\begin{equation}
\mathbf{v}\left(  \mathrm{i}s-0\right)  =\mathbf{v}\left(  \mathrm{i}%
s+0\right)  \left(
\begin{array}
[c]{cc}%
1 & -2\mathrm{i}\pi\chi_{-}\left(  s\right)  \delta\left(  s\right) \\
-2\mathrm{i}\pi\chi_{+}\left(  s\right)  \delta\left(  s\right)  & 1
\end{array}
\right)  , \label{jump cond across im line}%
\end{equation}
where%
\begin{equation}
\delta\left(  s\right)  :=\operatorname*{sgn}\left(  s\right)  \sum_{n}%
c_{n}^{2}\left(  t\right)  \delta_{\kappa_{n}}\left(  \left\vert s\right\vert
\right)  . \label{delta}%
\end{equation}
Observe that we have reformulated the pole conditions (\ref{pole cond}) as a
jump condition (\ref{jump cond across im line}) across the imaginary line, the
negative spectrum data being encoded in (\ref{delta}). Note that the jump
matrix in (\ref{jump cond across im line}) is not a continuous function but a
distribution. The main advantage of (\ref{jump cond across im line}) over
(\ref{pole cond}) is that it readily yields a generalization to an arbitrary
negative spectrum once we allow $\delta$ to be the distributional derivative
of an arbitrary positive finite measure. This will be done in Section
\ref{cont DT}.

\section{Quadrant-analytic vector Riemann-Hilbert problem\label{QARHP}}

\bigskip In this section we introduce a jump matrix with singular entries that
plays the central role in our considerations.

Let $\Sigma$ be a contour consisting of three lines:\ the real line
$\mathbb{R}$ oriented from left to right, the part of the imaginary line
$\mathrm{i}\mathbb{R}_{+}$ in the upper half-plane oriented upwards, and the
part of the imaginary line $-\mathrm{i}\mathbb{R}_{-}$ in the lower half-plane
oriented downwards. Apparently, $\Sigma$ divides the complex plane into
quadrants. Given a function $f\left(  k\right)  $ analytic on $\mathbb{C}%
\setminus\Sigma$, call such functions \emph{quadrant analytic\footnote{In
fact, $f$ may be analytic on some sets of the imaginary line. In particular,
may be analytic on both $\mathbb{C}^{\pm}.$}}, we denote by $f_{\pm}\left(
k\right)  $ nontangentional boundary values of $f$ from the positive/negative
$\left(  \pm\right)  $ side of $\Sigma$. Here the positive/negative side is
the one that lies to the left/right from $\Sigma$ as we traverse the contour
in the direction of orientation.

Let $R\left(  k\right)  $ be as in Proposition \ref{props of R}%
\[
R(k,t)=R(k)\mathrm{e}^{8\mathrm{i}k^{3}t},
\]
$\rho\left(  k\right)  $ a nonnegative finite measure on $\mathbb{R}_{+}$ and
$\delta\left(  k,t\right)  =$\textrm{e}$^{4k^{3}t}$d$\rho\left(  k\right)
/$d$k$ (generalized density).

\textbf{Quadrant analytic vector Riemann-Hilbert problem (QARHP)}%
\label{quadrant} Let $\boldsymbol{J}\left(  k,t\right)  $ be a $2\times2$
matrix-valued function defined on $\Sigma$ as follows%
\begin{equation}
\mathbf{J}\left(  k,t\right)  =\left\{
\begin{array}
[c]{cc}%
\mathbf{J}_{R}\left(  k,t\right)  , & \operatorname{Im}k=0\\
\mathbf{J}_{\rho}\left(  k,t\right)  , & \operatorname{Re}k=0
\end{array}
\right.  , \label{J}%
\end{equation}
where%
\[
\mathbf{J}_{R}\left(  k,t\right)  :=%
\begin{pmatrix}
1-|R(k,t)|^{2} & -\overline{R(k,t)}\\
R(k,t) & 1
\end{pmatrix}
,k\in\mathbb{R},
\]%
\[
\mathbf{J}_{\rho}\left(  k,t\right)  :=\left(
\begin{array}
[c]{cc}%
1 & -2\mathrm{i}\pi\chi_{-}\left(  s\right)  \delta\left(  s,t\right) \\
-2\mathrm{i}\pi\chi_{+}\left(  s\right)  \delta\left(  s,t\right)  & 1
\end{array}
\right)  ,\ k=\mathrm{i}s,\ s\in\mathbb{R}.
\]
Find a row function $\mathbf{v=}\left(
\begin{array}
[c]{ccc}%
\varphi & , & \psi
\end{array}
\right)  $ analytic in each quadrant such that:

\begin{enumerate}
\item \emph{Symmetry conditions: }%
\begin{equation}
\overline{\mathbf{v}(\overline{k})}=\mathbf{v}(-k)=\mathbf{v}(k)%
\begin{pmatrix}
0 & 1\\
1 & 0
\end{pmatrix}
. \label{sym}%
\end{equation}

\item \emph{Jump condition: }at least in the sense of distributions%
\begin{equation}
\mathbf{v}_{+}(k)=\mathbf{v}_{-}(k)\mathbf{J}(k,t),\text{ for }k\in\Sigma.
\label{jump}%
\end{equation}

\end{enumerate}

\begin{remark}
\begin{enumerate}
\item We stated our QARHP in abstract terms but it, of course, comes from the
Riemann-Hilbert formulation of the IST for the KdV equation. Thus, each
component of $\mathbf{v=}\left(
\begin{array}
[c]{ccc}%
\varphi & , & \psi
\end{array}
\right)  $ solves the Schr\"{o}dinger equation (and hence depends on the
spatial variable $x$) and the jump matrix $\mathbf{J}\left(  k,t\right)  $
depends on $t$ (time) as it takes into account the time evolution of
scattering data. Therefore, $\mathbf{v}\left(  k\right)  $ everywhere below
depends on $\left(  x,t\right)  $ as parameters (and consequently many other
quantities) but we agree to suppress this dependence when it causes no confusion.

\item We do not assume that $\boldsymbol{v}_{\pm}$ are continuous (which is
typically assumed). Since $\boldsymbol{J}$ is a distribution one then may
reasonably ask (as one of the referees did) "What does it mean to multiply a
distribution by a merely almost-everywhere defined function?" We see no way
such multiplication can be well-defined but there is no reason for that in our
situation. Indeed, while for a generic pair of solutions $\varphi,\psi$ to the
Schr\"{o}dinger equation condition 2 need not be satisfied in any sense, in
the context of (right) step-type potentials there is no problem with
(\ref{jump}). To explain this let us examine $\boldsymbol{J}_{R}$ and
$\boldsymbol{J}_{\rho}$. Since $\boldsymbol{J}_{R}$ is in $L^{\infty}$ there
is no problem on $\mathbb{R}$. On $\Sigma\diagdown\mathbb{R=}\mathrm{i}%
\mathbb{R}$ the jump matrix $\boldsymbol{J}_{\rho}$ has indeed distributional
off diagonal entries but once we rewrite (\ref{jump}) as ($s\geq0$)%
\begin{align*}
& \left(
\begin{array}
[c]{ccc}%
\varphi\left(  \mathrm{i}s-0\right)   & , & \psi\left(  \mathrm{i}s-0\right)
\end{array}
\right)  \\
& =\left(
\begin{array}
[c]{ccc}%
\varphi\left(  \mathrm{i}s+0\right)  -2\mathrm{i}\pi\delta\left(  s,t\right)
\psi\left(  \mathrm{i}s+0\right)   & , & \psi\left(  \mathrm{i}s+0\right)
\end{array}
\right)  ,
\end{align*}
one immediately sees that if $\psi$ is the right Jost solution (necessarily
analytic for $\operatorname{Im}k>0$ with continuous boundary values and hence
$\psi\left(  \mathrm{i}s-0\right)  =\psi\left(  \mathrm{i}s+0\right)
=\psi\left(  \mathrm{i}s\right)  $ is smooth) the product $\delta\left(
s,t\right)  \psi\left(  \mathrm{i}s\right)  $ is then a
distribution\footnote{Or even a function if the measure $\rho$ is an a.c.} for
$s\geq0$. Due to symmetry (\ref{sym}) this also holds for $s<0$. Thus the bad
situation when condition 2 is ill-posed (for nontrivial $\mathbf{v}$) may not
materialize for the type of problems we are concerned in this paper.

\item The asymptotic condition (\ref{asym conds}) imposes \cite{GT09} the
following condition on $\mathbf{J}$%
\[
\mathbf{J}(-k,t)=%
\begin{pmatrix}
0 & 1\\
1 & 0
\end{pmatrix}
\mathbf{J}(k,t)^{-1}%
\begin{pmatrix}
0 & 1\\
1 & 0
\end{pmatrix}
,
\]
which our $\mathbf{J}$ clearly obeys.

\item The asymptotic condition (\ref{asym conds}) is clearly missing which
results in non-uniqueness. Indeed, a scalar multiple of a solution satisfies
conditions (\ref{sym}) and (\ref{jump}). In our setting it is more convenient
to restore uniqueness by imposing a slightly different from (\ref{asym conds}) condition.
\end{enumerate}
\end{remark}

The following proposition is easy but crucial to our considerations. Since the
variables $\left(  x,t\right)  $ appear in the QARHP as parameters we leave
them out.

\begin{proposition}
[Gauge transformation]\label{prop on gauge}If $\mathbf{v}\left(  k\right)  $
is a solution of QARHP and $\mu$ is any finite (signed) symmetric measure,
then%
\begin{equation}
\widetilde{\mathbf{v}}\left(  k\right)  =\mathbf{v}\left(  k\right)
+\int\frac{W\left\{  \mathbf{v}\left(  k,\cdot\right)  ,\psi\left(
\cdot,\mathrm{i}s\right)  \right\}  }{k^{2}+s^{2}}\mathrm{d}\mu\left(
s\right)  ,\ k\notin\Sigma, \label{gauge}%
\end{equation}
satisfies the symmetry condition (\ref{sym}) and%
\[
\widetilde{\mathbf{v}}_{+}(k)=\widetilde{\mathbf{v}}_{-}(k)\mathbf{J}%
_{R}(k,t)\text{ }\operatorname{Im}k=0\text{ (jump across the real line),}%
\]
but not across the imaginary line.
\end{proposition}

\begin{proof}
Check the symmetry conditions (\ref{sym}):%
\begin{align*}
\overline{\widetilde{\mathbf{v}}(\overline{k})}  &  =\overline{\mathbf{v}%
(\overline{k})}+\int\frac{W\left\{  \overline{\mathbf{v}(\overline{k})}%
,\psi\right\}  }{k^{2}+s^{2}}\mathrm{d}\mu\\
&  =\mathbf{v}(-k)+\int\frac{W\left\{  \mathbf{v}(-k),\psi\right\}  }%
{k^{2}+s^{2}}\mathrm{d}\mu=\widetilde{\mathbf{v}}(-k)\\
&  =\mathbf{v}(k)%
\begin{pmatrix}
0 & 1\\
1 & 0
\end{pmatrix}
+\int\frac{W\left\{  \mathbf{v}(k)%
\begin{pmatrix}
0 & 1\\
1 & 0
\end{pmatrix}
,\psi\right\}  }{k^{2}+s^{2}}\mathrm{d}\mu\\
&  =\widetilde{\mathbf{v}}(k)%
\begin{pmatrix}
0 & 1\\
1 & 0
\end{pmatrix}
\end{align*}
and both symmetry conditions follow. Check the jump across condition
(\ref{jump}) across the real line. Suppressing the variables, we have%
\begin{align*}
\widetilde{\mathbf{v}}_{+}  &  =\mathbf{v}_{+}+\int\frac{W\left\{
\mathbf{v}_{+},\psi\right\}  }{k^{2}+s^{2}}\mathrm{d}\mu=\mathbf{v}_{-}\left(
k\right)  \mathbf{J}_{R}+\int\frac{W\left\{  \mathbf{v}_{-}\mathbf{J}_{R}%
,\psi\right\}  }{k^{2}+s^{2}}\mathrm{d}\mu\\
&  =\left(  \mathbf{v}_{-}+\int\frac{W\left\{  \mathbf{v}_{-},\psi\right\}
}{k^{2}+s^{2}}\mathrm{d}\mu\right)  \mathbf{J}_{R}\mathbf{=}%
\widetilde{\mathbf{v}}_{-}\mathbf{J}_{R}%
\end{align*}
and the jump condition follows.
\end{proof}

Proposition \ref{prop on gauge} says that the transformation $\mathbf{v}%
_{+}\rightarrow\widetilde{\mathbf{v}}_{+}$ preserves the jump condition across
the real line but not across the imaginary line. By choosing $\mu$ we will be
able to modify the jump matrix $\mathbf{J}_{\rho}$ in (\ref{J}) in nearly
unrestricted way (without altering the reflection coefficient).

\section{Continuous binary Darboux transformation\label{cont DT}}

In this section we state and prove a Riemann-Hilbert version of the continuous
binary Darboux transformation. It is of course directly related to Theorem
\ref{MainThm} but we hope it deserves special attention. Through this section
all statements and proofs admit time dependent situation but since $t$ appears
as just a parameter we drop it from the list of variables emphasizing that the
material of this section need not be considered in the KdV context.

\begin{theorem}
[Continuous Darboux transformation]\label{Darboux for RHP} Let $q\left(
x\right)  $ be a step-type potential, $\psi\left(  x,k\right)  $ a right Jost
solution, and for all real $x$%
\[
\mathbf{v}\left(  x,k\right)  \mathbf{=}\left(
\begin{array}
[c]{ccc}%
\varphi\left(  x,k\right)  & , & \psi\left(  x,k\right)
\end{array}
\right)
\]
solve the QARHP with the jump matrix $\left(  \mathbf{J}_{R}\mathbf{,J}_{\rho
}\right)  $ given by (\ref{J}). Let $\sigma$ be a (signed) finite measure
supported on $\mathbb{R}_{+}$ such that the Fredholm integral equation%
\begin{equation}
y\left(  \alpha,x\right)  +\int K\left(  \alpha,s,x\right)  y\left(
s,x\right)  \mathrm{d}\sigma\left(  s\right)  =\psi\left(  x,\mathrm{i}%
\alpha\right)  ,\ \ \ \alpha\in\operatorname*{Supp}\sigma, \label{singular eq}%
\end{equation}
with the kernel%
\[
K\left(  \alpha,s,x\right)  =\int_{x}^{\infty}\psi\left(  z,\mathrm{i}%
\alpha\right)  \psi\left(  z,\mathrm{i}s\right)  \mathrm{d}z
\]
has a unique solution in $L^{2}\left(  \mathrm{d}\sigma\right)  $. Then%
\begin{align}
\widetilde{\mathbf{v}}\left(  x,k\right)   &  \mathbf{=}\left(
\begin{array}
[c]{ccc}%
\widetilde{\varphi}\left(  x,k\right)  & , & \widetilde{\psi}\left(
x,k\right)
\end{array}
\right) \nonumber\\
&  =\mathbf{v}\left(  x,k\right)  +\int W\{\mathbf{v}\left(  x,k\right)
,\psi\left(  x,\mathrm{i}s\right)  \}\frac{y\left(  s,x\right)  \mathrm{d}%
\sigma\left(  s\right)  }{k^{2}+s^{2}} \label{RHP general soltn}%
\end{align}
solves the QARHP with the jump matrix $\left(  \mathbf{J}_{R}\mathbf{,J}%
_{\rho+\sigma}\right)  $. Moreover $\widetilde{\psi}\left(  x,k\right)  $ is
the right Jost solution corresponding to the potential%
\begin{align}
\widetilde{q}\left(  x,t\right)   &  =q\left(  x,t\right) \label{q tilda}\\
&  +2\left[  \int\widetilde{\psi}\left(  x;\mathrm{i}s\right)  \psi\left(
x,\mathrm{i}s\right)  \mathrm{d}\sigma\left(  s\right)  \right]  ^{2}%
+4\int\widetilde{\psi}\left(  x,\mathrm{i}s\right)  \psi^{\prime}\left(
x,\mathrm{i}s\right)  \mathrm{d}\sigma\left(  s\right)  ,\nonumber
\end{align}
and
\begin{equation}
\widetilde{\psi}\left(  x,\mathrm{i}\alpha\right)  =y\left(  \alpha,x\right)
,\ \ \ \alpha\in\operatorname*{Supp}\sigma. \label{y}%
\end{equation}

\end{theorem}

Note that the requirement that $\psi$ is a Jost solution plays the role of the
normalization condition missing in the QARHP. We also observe that since the
Jost solution is real for $\operatorname{Re}k=0$, the right hand side of
(\ref{q tilda}) is also real. And finally it is worth mentioning that
Condition (\ref{J}) implies that on the real line $\varphi\left(  x,k\right)
$ and $\psi\left(  x,k\right)  $ are related by the basic scattering identity
(\ref{eq8.1}) and thus $\varphi\left(  x,k\right)  $ is unique for a given the
latter $\psi\left(  x,k\right)  $.

\begin{proof}
Consider the gauge transform (\ref{gauge}). By Proposition \ref{prop on gauge}
the reflection coefficient $R$ is then preserved and it remains to find a
measure d$\mu\left(  s,x\right)  $ in (\ref{gauge}) that produces the
desirable jump matrix across the imaginary line. Due to the symmetry condition
(\ref{sym}) we can assume that $\operatorname{Im}k>0$. For the time being we
suppress the dependence on $x$ whenever it leads to no confusion.

Rewriting (\ref{gauge}) component-wise we have%
\begin{align}
\widetilde{\varphi}\left(  k\right)   &  =\varphi\left(  k\right)  +\int%
\frac{W\left\{  \varphi\left(  k\right)  ,\psi\left(  \mathrm{i}s\right)
\right\}  }{k^{2}+s^{2}}\mathrm{d}\mu\left(  s\right)  ,\label{fi tilda}\\
\widetilde{\psi}\left(  k\right)   &  =\psi\left(  k\right)  +\int%
\frac{W\left\{  \psi\left(  k\right)  ,\psi\left(  \mathrm{i}s\right)
\right\}  }{k^{2}+s^{2}}\mathrm{d}\mu\left(  s\right)  . \label{psi tilda}%
\end{align}
Since the new pair $\left(
\begin{array}
[c]{ccc}%
\widetilde{\varphi} & , & \widetilde{\psi}%
\end{array}
\right)  $ must satisfy the jump condition across $\mathrm{i}\mathbb{R}$, we
have%
\begin{align*}
&  \left(
\begin{array}
[c]{ccc}%
\widetilde{\varphi}\left(  \mathrm{i}s-0\right)  & , & \widetilde{\psi}\left(
\mathrm{i}s-0\right)
\end{array}
\right) \\
&  =\left(
\begin{array}
[c]{ccc}%
\widetilde{\varphi}\left(  \mathrm{i}s+0\right)  & , & \widetilde{\psi}\left(
\mathrm{i}s+0\right)
\end{array}
\right)  \left(
\begin{array}
[c]{cc}%
1 & -2\mathrm{i}\pi\chi_{-}\left(  s\right)  \widetilde{\delta}\left(
s\right) \\
-2\mathrm{i}\pi\chi_{+}\left(  s\right)  \widetilde{\delta}\left(  s\right)  &
1
\end{array}
\right)  ,
\end{align*}
where $\widetilde{\delta}\left(  s\right)  $ is the density of the perturbed
measure $\widetilde{\rho}=\rho+\sigma$, we conclude that $\widetilde{\psi
}\left(  k\right)  $ is analytic for $\operatorname{Im}k>0$ and%
\[
\operatorname{Im}\widetilde{\psi}\left(  \mathrm{i}\alpha-0\right)
=0,\ \ \operatorname{Im}\widetilde{\varphi}\left(  \mathrm{i}\alpha-0\right)
=-\pi\widetilde{\delta}\left(  \alpha\right)  \widetilde{\psi}\left(
\mathrm{i}\alpha\right)  .
\]
It follows from (\ref{fi tilda}) that%
\begin{align}
&  \operatorname{Im}\widetilde{\varphi}\left(  \mathrm{i}\alpha-0\right)
-\operatorname{Im}\varphi\left(  \mathrm{i}\alpha-0\right) \label{I1+I2}\\
&  =\lim_{\varepsilon\rightarrow0}\operatorname{Im}\int\frac{W\left\{
\varphi\left(  \mathrm{i}\alpha-\varepsilon\right)  ,\psi\left(
\mathrm{i}s\right)  \right\}  }{\left(  \mathrm{i}\alpha-\varepsilon\right)
^{2}+s^{2}}\mathrm{d}\mu\left(  s\right) \nonumber\\
&  =\lim_{\varepsilon\rightarrow0}\int W\left\{  \operatorname{Im}%
\varphi\left(  \mathrm{i}\alpha-\varepsilon\right)  ,\psi\left(
\mathrm{i}s\right)  \right\}  \operatorname{Re}\frac{1}{\left(  \mathrm{i}%
\alpha-\varepsilon\right)  ^{2}+s^{2}}\mathrm{d}\mu\left(  s\right)
\nonumber\\
&  +\lim_{\varepsilon\rightarrow0}\int W\left\{  \operatorname{Re}%
\varphi\left(  \mathrm{i}\alpha-\varepsilon\right)  ,\psi\left(
\mathrm{i}s\right)  \right\}  \operatorname{Im}\frac{1}{\left(  \mathrm{i}%
\alpha-\varepsilon\right)  ^{2}+s^{2}}\mathrm{d}\mu\left(  s\right)
\nonumber\\
&  =I_{1}+I_{2}.\nonumber
\end{align}
It follows from the jump condition for $\left(  \varphi,\psi\right)  $ that%
\begin{equation}
\operatorname{Im}\varphi\left(  \mathrm{i}\alpha-\varepsilon\right)
=-\pi\delta\left(  \alpha\right)  \psi\left(  \mathrm{i}\alpha\right)  .
\label{jump for fi}%
\end{equation}
Therefore, for $I_{1}$ we have%
\[
I_{1}=-\pi\delta\left(  \alpha\right)  \lim_{\varepsilon\rightarrow0}\int
W\left\{  \psi\left(  \mathrm{i}\alpha\right)  ,\psi\left(  \mathrm{i}%
s\right)  \right\}  \operatorname{Re}\frac{1}{\left(  \mathrm{i}%
\alpha-\varepsilon\right)  ^{2}+s^{2}}\mathrm{d}\mu\left(  s\right)  .
\]
Recall the following Wronskian identity: if $f_{\lambda}$ is a solution to the
Schr\"{o}dinger equation $-f^{\prime\prime}+q\left(  x\right)  f=\lambda^{2}f$
\ then%
\begin{equation}
W^{\prime}\left\{  f_{\lambda},f_{\nu}\right\}  =\left(  \lambda^{2}-\nu
^{2}\right)  f_{\lambda}f_{\nu}\text{ for any }\lambda,\nu. \label{W}%
\end{equation}
Observe that if $f_{\lambda}$, $f_{\nu}$ decay sufficiently fast at $+\infty$,
then (\ref{W}) implies%
\begin{equation}
\frac{W\left\{  f_{\lambda},f_{\nu}\right\}  }{\lambda^{2}-\nu^{2}}%
\mathbf{=-}\int_{x}^{\infty}f_{\lambda}\left(  s\right)  f_{\nu}\left(
s\right)  \mathrm{d}s. \label{W1}%
\end{equation}
Since due to (\ref{W1})%
\begin{align}
W\left\{  \psi\left(  \mathrm{i}\alpha\right)  ,\psi\left(  \mathrm{i}%
s\right)  \right\}   &  =-\left(  s^{2}-\alpha^{2}\right)  \int_{x}^{\infty
}\psi\left(  z,\mathrm{i}\alpha\right)  \psi\left(  z,\mathrm{i}s\right)
\mathrm{d}z\label{Wpsi}\\
&  =-\left(  s^{2}-\alpha^{2}\right)  K\left(  \alpha,s\right)  ,\nonumber
\end{align}
where we have denoted (suppressing $x$ as before)%
\[
K\left(  \alpha,s\right)  =K\left(  \alpha,s;x\right)  =\int_{x}^{\infty}%
\psi\left(  z,\mathrm{i}\alpha\right)  \psi\left(  z,\mathrm{i}s\right)
\mathrm{d}z,
\]
the last equation can be continued%
\[
I_{1}=\pi\delta\left(  \alpha\right)  \lim_{\varepsilon\rightarrow0}%
\int\left(  s^{2}-\alpha^{2}\right)  \operatorname{Re}\frac{1}{\left(
\mathrm{i}\alpha-\varepsilon\right)  ^{2}+s^{2}}K\left(  \alpha,s\right)
\mathrm{d}\mu\left(  s\right)  .
\]
Observe that%
\begin{align*}
\operatorname{Re}\frac{1}{\left(  \mathrm{i}\alpha-\varepsilon\right)
^{2}+s^{2}}  &  =\frac{1}{2s}\operatorname{Re}\left(  \frac{1}{s-\alpha
-\mathrm{i}\varepsilon}+\frac{1}{s+\alpha+\mathrm{i}\varepsilon}\right) \\
&  =\frac{1}{2s}\left[  \frac{s-\alpha}{\left(  s-\alpha\right)
^{2}+\varepsilon^{2}}+\frac{s+\alpha}{\left(  s+\alpha\right)  ^{2}%
+\varepsilon^{2}}\right]  ,
\end{align*}
and hence%
\begin{align*}
&  \left(  s^{2}-\alpha^{2}\right)  \operatorname{Re}\frac{1}{\left(
\mathrm{i}\alpha-\varepsilon\right)  ^{2}+s^{2}}\\
&  =\frac{1}{2s}\left[  \frac{\left(  s-\alpha\right)  ^{2}\left(
s+\alpha\right)  }{\left(  s-\alpha\right)  ^{2}+\varepsilon^{2}}%
+\frac{\left(  s+\alpha\right)  ^{2}\left(  s-\alpha\right)  }{\left(
s+\alpha\right)  ^{2}+\varepsilon^{2}}\right] \\
&  \rightarrow\frac{1}{2s}\left[  \left(  s+\alpha\right)  +\left(
s-\alpha\right)  \right]  =1,\ \ \ \varepsilon\rightarrow0\text{
\ \ uniformly.}%
\end{align*}
Therefore,%
\begin{equation}
I_{1}=\pi\delta\left(  \alpha\right)  \int K\left(  \alpha,s\right)
\mathrm{d}\mu\left(  s\right)  . \label{I1}%
\end{equation}
Turn to $I_{2}$. We have%
\begin{align*}
\operatorname{Im}\frac{1}{\left(  \mathrm{i}\alpha-\varepsilon\right)
^{2}+s^{2}}  &  =\frac{1}{2s}\operatorname{Im}\left(  \frac{1}{s-\alpha
-\mathrm{i}\varepsilon}+\frac{1}{s+\alpha+\mathrm{i}\varepsilon}\right) \\
&  =\frac{1}{2s}\left[  \frac{\varepsilon}{\left(  s-\alpha\right)
^{2}+\varepsilon^{2}}-\frac{\varepsilon}{\left(  s+\alpha\right)
^{2}+\varepsilon^{2}}\right] \\
&  =\frac{\pi}{2s}P_{\alpha+\mathrm{i}\varepsilon}\left(  s\right)  -\frac
{1}{2s}\frac{\varepsilon}{\left(  s+\alpha\right)  ^{2}+\varepsilon^{2}},
\end{align*}
where%
\[
P_{x+\mathrm{i}y}\left(  t\right)  =\dfrac{1}{\pi}\frac{y}{\left(  t-x\right)
^{2}+y^{2}}%
\]
is the the Poisson kernel. Recall the classical fact (see, e.g. \cite{Koosis})%
\[
\int\mathrm{d}xg\left(  x\right)  \lim_{y\rightarrow0}\int P_{x+\mathrm{i}%
y}\left(  t\right)  \mathrm{d}\mu\left(  t\right)  =\int g\left(  x\right)
\mathrm{d}\mu\left(  x\right)
\]
or%
\begin{equation}
\mathrm{d}\mu_{y}\left(  x\right)  =\left(  \int P_{x+\mathrm{i}y}\left(
t\right)  \mathrm{d}\mu\left(  t\right)  \right)  \mathrm{d}x\rightarrow
\mathrm{d}\mu\left(  x\right)  ,y\rightarrow0, \label{weak}%
\end{equation}
in the weak* sense. The second term goes to zero uniformly as $\varepsilon
\rightarrow0$ and we get%
\begin{align*}
I_{2}  &  =\lim_{\varepsilon\rightarrow0}\int W\left\{  \operatorname{Re}%
\varphi\left(  \mathrm{i}\alpha-\varepsilon\right)  ,\psi\left(
\mathrm{i}s\right)  \right\}  \operatorname{Im}\frac{1}{\left(  \mathrm{i}%
\alpha-\varepsilon\right)  ^{2}+s^{2}}\mathrm{d}\mu\left(  s\right) \\
&  =\lim_{\varepsilon\rightarrow0}\int W\left\{  \operatorname{Re}%
\varphi\left(  \mathrm{i}\alpha-0\right)  ,\psi\left(  \mathrm{i}s\right)
\right\}  \frac{\pi}{2s}P_{\alpha+\mathrm{i}\varepsilon}\left(  s\right)
\mathrm{d}\mu\left(  s\right) \\
&  =\pi\int\frac{W\left\{  \operatorname{Re}\varphi\left(  \mathrm{i}%
\alpha-0\right)  ,\psi\left(  \mathrm{i}\alpha\right)  \right\}  }{2\alpha
}\mathrm{d}\mu\left(  \alpha\right)  .
\end{align*}
Here we have used (\ref{weak}) to pass to the limit.

It follows from (\ref{jump}) that
\[
\varphi\left(  k+\mathrm{i}0\right)  =\overline{\psi\left(  k+\mathrm{i}%
0\right)  }+R\left(  k\right)  \psi\left(  k+\mathrm{i}0\right)
\]
and hence, since $\psi$ is the right Jost solution,
\[
W\left\{  \varphi\left(  k+\mathrm{i}0\right)  ,\psi\left(  k+\mathrm{i}%
0\right)  \right\}  =W\left\{  \overline{\psi\left(  k+\mathrm{i}0\right)
},\psi\left(  k+\mathrm{i}0\right)  \right\}  =2\mathrm{i}k.
\]
Thus $W\left\{  \varphi\left(  k\right)  ,\psi\left(  k\right)  \right\}
=2\mathrm{i}k$ also for $\operatorname{Im}k>0$ and in particular
\begin{align*}
&  W\left\{  \operatorname{Re}\varphi\left(  \mathrm{i}\alpha-0\right)
,\psi\left(  \mathrm{i}\alpha\right)  \right\} \\
&  =\operatorname{Re}W\left\{  \varphi\left(  \mathrm{i}\alpha-0\right)
,\psi\left(  \mathrm{i}\alpha\right)  \right\}  =-2\alpha.
\end{align*}
It follow then that in the sense of distributions%
\begin{equation}
I_{2}=-\pi\mu^{\prime}\left(  \alpha\right)  . \label{I2}%
\end{equation}
Substituting (\ref{I1}) and (\ref{I2}) into (\ref{I1+I2}) yields%
\begin{align*}
\operatorname{Im}\widetilde{\varphi}\left(  \mathrm{i}\alpha-0\right)   &
=\operatorname{Im}\varphi\left(  \mathrm{i}\alpha-0\right)  +I_{1}+I_{2}\\
&  =-\delta\left(  \alpha\right)  \psi\left(  \mathrm{i}\alpha\right)
+\pi\delta\left(  \alpha\right)  \int K\left(  \alpha,s\right)  \mathrm{d}%
\mu\left(  s\right)  -\pi\mu^{\prime}\left(  \alpha\right) \\
&  =-\pi\widetilde{\delta}\left(  \alpha\right)  \widetilde{\psi}\left(
\mathrm{i}\alpha\right)  .
\end{align*}
Here we have taken (\ref{jump for fi}) into account. It follows that%
\begin{equation}
\delta\left(  \alpha\right)  \left[  \psi\left(  \mathrm{i}\alpha\right)
-\int K\left(  \alpha,s\right)  \mathrm{d}\mu\left(  s\right)  \right]
+\mu^{\prime}\left(  \alpha\right)  =\widetilde{\delta}\left(  \alpha\right)
\widetilde{\psi}\left(  \mathrm{i}\alpha\right)  . \label{jump for fi 2}%
\end{equation}
Evaluate now $\widetilde{\psi}\left(  \mathrm{i}\alpha\right)  $. From
(\ref{psi tilda}) and (\ref{Wpsi}) we see that%
\begin{align*}
\widetilde{\psi}\left(  k\right)   &  =\psi\left(  k\right)  +\int%
\frac{W\left\{  \psi\left(  k\right)  ,\psi\left(  \mathrm{i}s\right)
\right\}  }{k^{2}+s^{2}}\mathrm{d}\mu\left(  s\right) \\
&  =\psi\left(  k\right)  -\int K\left(  \alpha,s\right)  \mathrm{d}\mu\left(
s\right)
\end{align*}
and hence%
\begin{equation}
\widetilde{\psi}\left(  \mathrm{i}\alpha\right)  =\psi\left(  \mathrm{i}%
\alpha\right)  -\int K\left(  \alpha,s\right)  \mathrm{d}\mu\left(  s\right)
. \label{psi tilda 1}%
\end{equation}
Substituting this into (\ref{jump for fi 2}) one immediately obtains%
\begin{align*}
&  \delta\left(  \alpha\right)  \left[  \psi\left(  \mathrm{i}\alpha\right)
-\int K\left(  \alpha,s\right)  \mathrm{d}\mu\left(  s\right)  \right]
+\mu^{\prime}\left(  \alpha\right) \\
&  =\widetilde{\delta}\left(  \alpha\right)  \left\{  \psi\left(
\mathrm{i}\alpha\right)  -\int K\left(  \alpha,s\right)  \mathrm{d}\mu\left(
s\right)  \right\}  ,
\end{align*}
which, recalling that $\sigma=\widetilde{\rho}-\rho$, can be rearranged as%
\[
\sigma^{\prime}\left(  \alpha\right)  \left\{  \psi\left(  \mathrm{i}%
\alpha\right)  -\int K\left(  \alpha,s\right)  \mathrm{d}\mu\left(  s\right)
\right\}  =\mu^{\prime}\left(  \alpha\right)  ,
\]
or in terms of measures%
\[
\left\{  \psi\left(  \mathrm{i}\alpha\right)  -\int K\left(  \alpha,s\right)
\mathrm{d}\mu\left(  s\right)  \right\}  \mathrm{d}\sigma\left(  s\right)
=\mathrm{d}\mu\left(  \alpha\right)  ,
\]
which is an integral equation on the measure $\mu$. Take d$\mu$ to be
absolutely continuous with respect to the measure d$\sigma$ and let $y=$%
d$\mu/$d$\sigma$ be the Radon-Nikodym derivative. By the Radon-Nikodym theorem
then%
\[
\psi\left(  \mathrm{i}\alpha\right)  -\int K\left(  \alpha,s\right)  y\left(
s\right)  \mathrm{d}\sigma\left(  s\right)  =y\left(  \alpha\right)
\]
and we finally have%
\[
y\left(  \alpha,x\right)  +\int K\left(  \alpha,s,x\right)  y\left(
s,x\right)  \mathrm{d}\sigma\left(  s\right)  =\psi\left(  x,\mathrm{i}%
\alpha\right)  .
\]
Thus we arrive at the Fredholm integral equation%
\begin{equation}
y+\mathbb{K}y=\psi, \label{fredholm v1}%
\end{equation}
where $\mathbb{K}$ is the integral operator with the kernel $K\left(
\alpha,s;x\right)  $ with respect to the measure $\sigma$.

It remains to show (\ref{y}). This immediately follows from (\ref{psi tilda 1}%
) and (\ref{fredholm v1}). Indeed,%
\begin{align*}
\widetilde{\psi}\left(  \mathrm{i}\alpha\right)   &  =\psi\left(
\mathrm{i}\alpha\right)  -\int K\left(  \alpha,s\right)  \mathrm{d}\mu\left(
s\right) \\
&  =\psi\left(  \mathrm{i}\alpha\right)  -\int K\left(  \alpha,s\right)
y\left(  s\right)  \mathrm{d}\sigma\left(  s\right) \\
&  =\psi-\mathbb{K}y=y.
\end{align*}
The proof that $\widetilde{\psi}$ is a right Jost solution and the
representation (\ref{q tilda}) for $\widetilde{q}$ will be given in Section
\ref{proof of main theorem}.
\end{proof}

\section{An integral operator and its discretization\label{discrete oper}}

In this section we study the trace class integral operator $\mathbb{K}$
arising in the previous section and approximate it with a sequence of finite
matrices. It is convenient to do so in independent terms. Let $\mu$ be a
non-negative finite Borel measure on the real line and $L^{2}\left(
\mathrm{d}\mu\right)  $ the real Hilbert space with the inner product
$\left\langle f,g\right\rangle =\int f\left(  x\right)  g\left(  x\right)
$d$\mu\left(  x\right)  $. Let $g\left(  x,s\right)  $ be a real continuous
function for $s\in S$, where $S$ is an interval (finite or infinite) such that%
\begin{equation}
\left\vert \left\vert \left\vert g\right\vert \right\vert \right\vert
^{2}:=\int_{S}\int g\left(  x,s\right)  ^{2}\mathrm{d}\mu\left(  x\right)
\mathrm{d}s<\infty. \label{g}%
\end{equation}
Define a family of rank one operators $\mathbb{G}\left(  s\right)  $ on
$L^{2}\left(  \mathrm{d}\mu\right)  $ by%
\[
\left(  \mathbb{G}\left(  s\right)  f\right)  \left(  x\right)  =\left\langle
f,g\left(  \cdot,s\right)  \right\rangle g\left(  x,s\right)  .
\]
Clearly, $\mathbb{G}$ is positive and ($\left\vert \left\vert \cdot\right\vert
\right\vert _{2}$ stands for the Hilbert-Schmidt norm)%
\begin{equation}
\operatorname*{tr}\mathbb{G}\left(  s\right)  =\left\vert \left\vert
\mathbb{G}\left(  s\right)  \right\vert \right\vert _{2}=\left\vert \left\vert
g\left(  \cdot,s\right)  \right\vert \right\vert ^{2}=\int g\left(
x,s\right)  ^{2}\mathrm{d}\mu\left(  x\right)  . \label{trG}%
\end{equation}
Consider an operator defined by%
\[
\mathbb{K=}\int_{S}\mathbb{G}\left(  s\right)  \mathrm{d}s.
\]
It is an integral operator%
\[
\left(  \mathbb{K}f\right)  \left(  x\right)  =\int K\left(  x,y\right)
f\left(  y\right)  \mathrm{d}\mu\left(  y\right)
\]
on $L^{2}\left(  \mathrm{d}\mu\right)  $ with the kernel%
\[
K\left(  x,y\right)  =\int_{S}g\left(  x,s\right)  g\left(  y,s\right)
\mathrm{d}s.
\]
It follows from (\ref{g}) and (\ref{trG}) that%
\[
\left\vert \left\vert \mathbb{K}\right\vert \right\vert _{2}\leq\int%
_{S}\left\vert \left\vert \mathbb{G}\left(  s\right)  \right\vert \right\vert
_{2}\mathrm{d}s=\int_{S}\int g\left(  x,s\right)  ^{2}\mathrm{d}\mu\left(
x\right)  \mathrm{d}s=\left\vert \left\vert \left\vert g\right\vert
\right\vert \right\vert ^{2}<\infty
\]
and hence the operator $\mathbb{K}$ is Hilbert-Schmidt on $L^{2}\left(
\mathrm{d}\mu\right)  $, positive, it is also trace class, and%
\[
\operatorname*{tr}\mathbb{K}=\int_{S}\int g\left(  x,s\right)  ^{2}%
\mathrm{d}\mu\left(  x\right)  \mathrm{d}s=\left\vert \left\vert \left\vert
g\right\vert \right\vert \right\vert ^{2}.
\]
Discretize $\mathbb{K}$ as follows. Let $I$ be a finite interval containing
$\operatorname*{Supp}\mu$ and $\left(  I_{n}\right)  _{n=1}^{N}$ be a finite
partition of $I$. In each $I_{n}$ pick up an interior point $x_{n}$ that is
also in $\operatorname*{Supp}\mu$. Take a piece-wise constant approximation
$g_{N}$ of $g$%
\[
g_{N}\left(  x,s\right)  =%
{\displaystyle\sum\limits_{n=1}^{N}}
g\left(  x_{n},s\right)  \chi_{I_{n}}\left(  x\right)
\]
and consider the integral operator%
\[
\left(  \mathbb{K}_{N}f\right)  \left(  x\right)  =\int K_{N}\left(
x,y\right)  f\left(  y\right)  \mathrm{d}\mu\left(  y\right)
\]
on $L^{2}\left(  \mathrm{d}\mu\right)  $ with the kernel%
\[
K_{N}\left(  x,y\right)  =\int_{S}g_{N}\left(  x,s\right)  g_{N}\left(
y,s\right)  \mathrm{d}s.
\]

\begin{proposition}
\label{Prop on limit}If $\left\vert \left\vert \left\vert g-g_{N}\right\vert
\right\vert \right\vert \rightarrow0,N\rightarrow\infty$, then $\mathbb{K}%
_{N}\rightarrow\mathbb{K}$ in the trace norm.
\end{proposition}

The proof is given in the Appendix. We show that $\mathbb{K}_{N}$ can be
realized as the $N\times N$ matrix%
\[
\boldsymbol{K}_{N}=\left(  c_{n}^{2}\int_{S}g_{m}\left(  s\right)
g_{n}\left(  s\right)  \mathrm{d}s\right)  _{1\leq m,n\leq N},
\]
where $g_{n}\left(  s\right)  :=g\left(  x_{n},s\right)  $ and $c_{n}%
:=\mu\left(  I_{n}\right)  ^{1/2}$, as follows. Identify a simple function%
\[
f_{N}\left(  x\right)  =%
{\displaystyle\sum\limits_{n=1}^{N}}
f_{n}\chi_{I_{n}}\left(  x\right)
\]
with an $N$ column $\left(  f_{n}\right)  :=\boldsymbol{f}_{N}$. We show that%
\[
\left.  \mathbb{K}_{N}f_{N}\right\vert _{I_{n}}=\left(  \boldsymbol{K}%
_{N}\boldsymbol{f}_{N}\right)  _{n},
\]
where\ the subscript $n$ denotes the $n$the component of a column. Indeed,%
\begin{align*}
\left(  \mathbb{K}_{N}f_{N}\right)  \left(  x\right)   &  =\int K_{N}\left(
x,y\right)  f_{N}\left(  y\right)  \mathrm{d}\mu\left(  y\right) \\
&  =\int\left[  \int_{S}g_{N}\left(  x,s\right)  g_{N}\left(  y,s\right)
\mathrm{d}s\right]  f_{N}\left(  y\right)  \mathrm{d}\mu\left(  y\right) \\
&  =\int_{S}g_{N}\left(  x,s\right)  \left[  \int g_{N}\left(  y,s\right)
f_{N}\left(  y\right)  \mathrm{d}\mu\left(  y\right)  \right]  \mathrm{d}s\\
&  =\int_{S}g_{N}\left(  x,s\right)  \left[
{\displaystyle\sum\limits_{m=1}^{N}}
g_{N}\left(  x_{m},s\right)  f_{m}\mu\left(  I_{m}\right)  \right]
\mathrm{d}s\\
&  =\int_{S}g_{N}\left(  x,s\right)  \left[
{\displaystyle\sum\limits_{m=1}^{N}}
c_{m}^{2}g_{m}\left(  s\right)  f_{m}\right]  \mathrm{d}s\\
&  =%
{\displaystyle\sum\limits_{m=1}^{N}}
c_{m}^{2}\left[  \int_{S}g_{N}\left(  x,s\right)  g_{m}\left(  s\right)
\mathrm{d}s\right]  f_{m}.
\end{align*}
It follows that for the restriction to $I_{n}$ we have%
\begin{align*}
\left.  \mathbb{K}_{N}f_{N}\right\vert _{I_{n}}  &  =%
{\displaystyle\sum\limits_{m=1}^{N}}
c_{m}^{2}\left[  \int_{S}g_{n}\left(  s\right)  g_{m}\left(  s\right)
\mathrm{d}s\right]  f_{m}\\
&  =\left(  \boldsymbol{K}_{N}\boldsymbol{f}_{N}\right)  _{n}%
\end{align*}
as desired.

Thus we have shown that $\left(  I+\mathbb{K}_{N}\right)  ^{-1}$ is
well-defined on $L^{2}\left(  \mathrm{d}\mu\right)  $ and the integral
equation%
\[
\left(  I+\mathbb{K}_{N}\right)  f=g_{N}%
\]
has a unique $L^{2}\left(  \mathrm{d}\mu\right)  $ solution%
\[
f_{N}=\left(  I+\mathbb{K}_{N}\right)  ^{-1}g_{N},
\]
which is due to Proposition \ref{Prop on limit} converges in $L^{2}\left(
\mathrm{d}\mu\right)  $ to some $f$.

\section{Abstract log-determinant formula\label{log-det}}

In this section we derive a log-determinant formula needed in the proof of
Theorem \ref{MainThm}. It is convenient to do so in independent terms.

\begin{proposition}
\label{prop on log-det}Let $\mathbb{A}\left(  x\right)  $ be a self-adjoint
trace class operator-valued function on a real Hilbert space $\mathfrak{H}$
such that%
\[
\partial_{x}\mathbb{A}\left(  x\right)  =-\left\langle \cdot,a\left(
x\right)  \right\rangle a\left(  x\right)  ,a\in\mathfrak{H}%
\]
is rank-one. If $\det\left\{  I+\mathbb{A}\left(  x\right)  \right\}  >0$ then%
\begin{align}
\partial_{x}^{2}\log\det\left[  I+\mathbb{A}\left(  x\right)  \right]   &
=-\left\langle a\left(  x\right)  ,\left[  I+\mathbb{A}\left(  x\right)
\right]  ^{-1}a\left(  x\right)  \right\rangle ^{2}%
\label{abstract log det formula}\\
&  -2\left\langle \partial_{x}a\left(  x\right)  ,\left[  I+\mathbb{A}\left(
x\right)  \right]  ^{-1}a\left(  x\right)  \right\rangle .\nonumber
\end{align}

\end{proposition}

\begin{proof}
We suppress dependence on $x$. We base our proof on the followings well-known
formulas (see e.g. \cite{SimonTraceIdeals})%
\begin{equation}
\log\det\left(  I+\mathbb{A}\right)  =\operatorname*{tr}\log\left(
I+\mathbb{A}\right)  , \label{tr1}%
\end{equation}%
\begin{equation}
\partial\operatorname*{tr}\log\left(  I+\mathbb{A}\right)  =\operatorname*{tr}%
\left(  I+\mathbb{A}\right)  ^{-1}\partial\mathbb{A}, \label{tr2}%
\end{equation}%
\begin{equation}
\operatorname*{tr}\left\langle \cdot,f\right\rangle g=\left\langle
g,f\right\rangle . \label{tr4}%
\end{equation}%
\begin{equation}
\partial\left(  I+\mathbb{A}\right)  ^{-1}=-\left(  I+\mathbb{A}\right)
^{-1}\partial\mathbb{A}\left(  I+\mathbb{A}\right)  ^{-1}. \label{tr3}%
\end{equation}
By (\ref{tr1}) and (\ref{tr2}) we have%
\begin{align*}
\partial^{2}\log\det\left(  I+\mathbb{A}\right)   &  =\partial
\operatorname*{tr}\left(  I+\mathbb{A}\right)  ^{-1}\partial\mathbb{A}\\
&  =\operatorname*{tr}\left(  I+\mathbb{A}\right)  ^{-1}\partial\mathbb{A}.
\end{align*}
But $\left(  I+\mathbb{A}\right)  ^{-1}\partial\mathbb{A}$ is a rank one
operator and hence by (\ref{tr4})%
\begin{align*}
\operatorname*{tr}\left(  I+\mathbb{A}\right)  ^{-1}\partial\mathbb{A}  &
=-\operatorname*{tr}\left\langle \cdot,a\right\rangle \left(  I+\mathbb{A}%
\right)  ^{-1}a\\
&  =-\left\langle \left(  I+\mathbb{A}\right)  ^{-1}a,a\right\rangle .
\end{align*}
Differentiating this equation one more time, by (\ref{tr3}) we have%
\begin{align}
\partial\operatorname*{tr}\left(  I+\mathbb{A}\right)  ^{-1}\partial
\mathbb{A}  &  =-\partial\left\langle \left(  I+\mathbb{A}\right)
^{-1}a,a\right\rangle =-\left\langle \partial\left(  I+\mathbb{A}\right)
^{-1}a,a\right\rangle \label{eval}\\
&  -\left\langle \left(  I+\mathbb{A}\right)  ^{-1}\partial a,a\right\rangle
-\left\langle \left(  I+\mathbb{A}\right)  ^{-1}a,\partial a\right\rangle
\nonumber\\
&  =\left\langle \left(  I+\mathbb{A}\right)  ^{-1}\partial\mathbb{A}\left(
I+\mathbb{A}\right)  ^{-1},a\right\rangle -2\left\langle \left(
I+\mathbb{A}\right)  ^{-1}a,\partial a\right\rangle \nonumber\\
&  =\left\langle \partial\mathbb{A}\left(  I+\mathbb{A}\right)  ^{-1},\left(
I+\mathbb{A}\right)  ^{-1}a\right\rangle -2\left\langle \left(  I+\mathbb{A}%
\right)  ^{-1}a,\partial a\right\rangle .\nonumber
\end{align}
But $\partial\mathbb{A}\left(  I+\mathbb{A}\right)  ^{-1}$ is a rank one
operator and therefore%
\[
\partial\mathbb{A}\left(  I+\mathbb{A}\right)  ^{-1}\cdot=-\left\langle
\cdot,\left(  I+\mathbb{A}\right)  ^{-1}a\right\rangle a.
\]
Thus, we have%
\begin{align*}
\left\langle \partial\mathbb{A}\left(  I+\mathbb{A}\right)  ^{-1}a,\left(
I+\mathbb{A}\right)  ^{-1}a\right\rangle  &  =-\left\langle a,\left(
I+\mathbb{A}\right)  ^{-1}a\right\rangle \left\langle a,\left(  I+\mathbb{A}%
\right)  ^{-1}a\right\rangle \\
&  =-\left\langle a,\left(  I+\mathbb{A}\right)  ^{-1}a\right\rangle ^{2}.
\end{align*}
Substituting this into (\ref{eval}) finally yields
(\ref{abstract log det formula}).
\end{proof}

\section{Proof of the main theorem\label{proof of main theorem}}

\bigskip The proof of Theorem \ref{MainThm} amounts to combining the
ingredients prepared above and is based in part on limiting arguments. We
start out from the following statement.

\begin{proposition}
[Adding/removing bound states]\label{finite neg spec} Let $q\left(
x,t\right)  $ be a step-type KdV solution with the scattering data
$S_{q}=\left\{  R,\mathrm{d}\rho\right\}  $ and $\psi\left(  x,t,k\right)  $
the right Jost solution. Fix the discrete measure%
\[
\mathrm{d}\sigma\left(  k\right)  =%
{\displaystyle\sum\limits_{n=1}^{N}}
c_{n}^{2}\delta_{\kappa_{n}}\left(  k\right)  \mathrm{d}k,\ \ \ N<\infty,
\]
with some positive $c_{n}^{2}$, $\kappa_{n}$, and introduce the $N\times N$
matrix function $\mathbf{K}\left(  x,t\right)  $ with entries%
\[
K_{mn}\left(  x,t\right)  =c_{n}^{2}e^{8\kappa_{n}^{3}t}\int_{x}^{\infty}%
\psi\left(  s,t,\mathrm{i}\kappa_{m}\right)  \psi\left(  s,t,\mathrm{i}%
\kappa_{n}\right)  \mathrm{d}s.
\]
Suppose that d$\rho$ is discrete with a finitely many pure points then

\begin{enumerate}
\item \label{part 1}the matrix $I+\mathbf{K}\left(  x,t\right)  $ is
nonsingular,%
\begin{equation}
q_{\sigma}\left(  x,t\right)  =q\left(  x,t\right)  -2\partial_{x}^{2}\log
\det\left\{  I+\mathbf{K}\left(  x,t\right)  \right\}  \label{q}%
\end{equation}
is also a step-type KdV solution with the scattering data $\left\{
R,\mathrm{d}\rho+\mathrm{d}\sigma\right\}  $, and%
\begin{align}
\psi_{\sigma}\left(  x,t,k\right)   &  =\psi\left(  x,t,k\right) \label{ksi}\\
&  -%
{\displaystyle\sum\limits_{n=1}^{N}}
c_{n}^{2}e^{8\kappa_{n}^{3}t}y_{n}\left(  x,t\right)  \int_{x}^{\infty}%
\psi\left(  s,t,k\right)  \psi\left(  s,t,\mathrm{i}\kappa_{n}\right)
\mathrm{d}s\nonumber
\end{align}
is the associated right Jost solution, where the column $\mathbf{y}=\left(
y_{n}\right)  $ is given by%
\[
\mathbf{y}\left(  x,t\right)  =\left\{  I+\mathbf{K}\left(  x,t\right)
\right\}  ^{-1}\mathbf{\psi}\left(  x,t\right)  ,\ \ \mathbf{\psi}\left(
x,t\right)  :=\left(  \psi\left(  x,t,\mathrm{i}\kappa_{n}\right)  \right)  ;
\]

\item \label{part 2}If d$\sigma\leq$d$\rho$, then $I-\mathbf{K}\left(
x,t\right)  $ is nonsingular,%
\begin{equation}
q_{-\sigma}\left(  x,t\right)  =q\left(  x,t\right)  -2\partial_{x}^{2}%
\log\det\left\{  I-\mathbf{K}\left(  x,t\right)  \right\}  \label{q-}%
\end{equation}
is also a step-type KdV solution with the scattering data $\left\{
R,\mathrm{d}\rho-\mathrm{d}\sigma\right\}  $, and%
\begin{align}
\psi_{-\sigma}\left(  x,t,k\right)   &  =\psi\left(  x,t,k\right)
\label{ksi-}\\
&  +%
{\displaystyle\sum\limits_{n=1}^{N}}
c_{n}^{2}e^{8\kappa_{n}^{3}t}y_{n}\left(  x,t\right)  \int_{x}^{\infty}%
\psi\left(  s,t,k\right)  \psi\left(  s,t,\mathrm{i}\kappa_{n}\right)
\mathrm{d}s\nonumber
\end{align}
is the associated right Jost solution, where the column $\mathbf{y}=\left(
y_{n}\right)  $ is given by%
\[
\mathbf{y}\left(  x,t\right)  =\left\{  I-\mathbf{K}\left(  x,t\right)
\right\}  ^{-1}\boldsymbol{\psi}\left(  x,t\right)  ,\ \ \boldsymbol{\psi
}\left(  x,t\right)  :=\left(  \psi\left(  x,t,\mathrm{i}\kappa_{n}\right)
\right)  ;
\]

\item The binary Darboux transform is invertible in the following sense%
\[
\left(  q_{\sigma}\right)  _{-\sigma}=q,\ \ \ \left(  \psi_{\sigma}\right)
_{-\sigma}=\psi.
\]

\end{enumerate}
\end{proposition}

This statement is not new. Part 1 follows from \cite[Theorem 4.1]%
{Gesztesyetal96} where it is proven in the most general spectral situation and
for arbitrary Sturm-Liouville operators but not in the IST context. Part 2
(removing eigenvalues) and 3 are not explicitly addressed therein but it can
of course be done along the same lines readily suggested in
\cite{Gesztesyetal96}. For short-range $q$'s both parts are proven in our
\cite{Ryb21} by completely different methods and in the IST\ context. The main
difference in the approaches is that in \cite{Gesztesyetal96} is obtained by
adding one eigenvalue at a time while in \cite{Ryb21} all eigenvalues are
added simultaneously. The formulation of Proposition \ref{finite neg spec} is,
however, new. Therefore we only show here that the operator $I-\mathbf{K}%
\left(  x,t\right)  $ is nonsingular ($I+\mathbf{K}\left(  x,t\right)  $
clearly is) and $\det\left\{  I-\mathbf{K}\left(  x,t\right)  \right\}  >0$
for all real $x$ and $t\geq0$. To this end we observe that%
\begin{align*}
\det\left\{  I-\mathbf{K}\left(  x,t\right)  \right\}   &  =\det\left\{
I-\boldsymbol{C}^{2}\int_{x}^{\infty}\boldsymbol{\psi}\left(  s,t\right)
^{T}\boldsymbol{\psi}\left(  s,t\right)  \mathrm{d}s\right\}  \\
&  =\det\left\{  I-\int_{x}^{\infty}\left(  \boldsymbol{\psi}\left(
s,t\right)  \boldsymbol{C}\right)  ^{T}\left(  \boldsymbol{\psi}\left(
s,t\right)  \boldsymbol{C}\right)  \mathrm{d}s\right\}  ,
\end{align*}
where $\boldsymbol{C}=\operatorname*{diag}\left(  c_{n}e^{4\kappa_{n}^{3}%
t}\right)  $. Since $\left(  c_{n}e^{4\kappa_{n}^{3}t}\psi\left(
x,t,\mathrm{i}\kappa_{n}\right)  \right)  _{n=1}^{N}$ is an orthonormal set in
$L^{2}$ it follows that%
\[
I-\int_{x}^{\infty}\left(  \boldsymbol{\psi}\left(  s,t\right)  \boldsymbol{C}%
\right)  ^{T}\left(  \boldsymbol{\psi}\left(  s,t\right)  \boldsymbol{C}%
\right)  \ \mathrm{d}s=\int_{-\infty}^{x}\left(  \boldsymbol{\psi}\left(
s,t\right)  \boldsymbol{C}\right)  ^{T}\left(  \boldsymbol{\psi}\left(
s,t\right)  \boldsymbol{C}\right)  \mathrm{d}s,
\]
which is clearly a positive definite matrix and the conclusion immediately follows.

Note that Proposition \ref{finite neg spec} can actually be independently
derived from \cite[Theorem 3.1]{Ryb21} by the following limiting arguments.
Consider the sequence $q_{n}=q\chi_{\left(  -n,\infty\right)  }$. As is proven
in \cite{GruRybSIMA15}, $S_{q_{n}}\rightarrow S_{q}$ weakly and the
corresponding sequence $q_{n}\left(  x,t\right)  $ converges point-wise to
some $q\left(  x,t\right)  $ for every $t>0$. An explicit formula for
$q\left(  x,t\right)  $ via $S_{q}$ is given in \cite{GruRybSIMA15} in terms
of Hankel operators. The corresponding sequence $\psi_{n}\left(  k;x,t\right)
$ converges uniformly to $\psi\left(  x,t;k\right)  $ on compact in
$\operatorname{Im}k>0$ (this is a general classical fact that holds even for
Weyl solutions). Part 3 for $q_{n}$ follows directly from Marchenko's
characterization of the inverse scattering problem for short-range potentials
\cite{MarchBook2011}. Indeed, his characterization of the scattering data
$S_{q}=\left\{  R,\mathrm{d}\rho\right\}  $ imposes no other condition on
d$\rho$ but $N<\infty$. The passage to the limit as $n\rightarrow\infty$
should then be in order.

Observe that Proposition \ref{finite neg spec} is a particular case of Theorem
\ref{MainThm}. However an independent proof of Theorem \ref{MainThm} involves
cumbersome technicalities which can be avoided by taking limits in already
known results.

We now proceed to the proof of Theorem \ref{MainThm}. As before, since
$\left(  x,t\right)  $ appear as parameters, we may drop them from the list of
variable whenever it is convenient and leads to no confusion. We recall that
prime stands for the derivative in $x$ (no derivatives in other variable
appear). It is sufficiently to show that Theorem \ref{MainThm} holds for Case
1: $S_{q}=\left\{  R,0\right\}  $, d$\sigma\geq0$ and Case 2: $S_{q}=\left\{
R,\mathrm{d}\rho\right\}  $, $-$d$\sigma\leq0$ where d$\sigma$ is a
restriction of the (non-negative) measure d$\rho$, as the general case is a
combination of the two. We concentrate on the proof for Case 1 since Case 2
relies on the very same arguments.

Case 1 (adding negative spectrum). Discretize the measure d$\rho$ in
(\ref{S_q}) as is done in Section \ref{discrete oper}. More specifically, take%
\[
\mathrm{d}\rho_{N}\left(  k\right)  =%
{\displaystyle\sum\limits_{n=1}^{N}}
\rho\left(  I_{n}\right)  \delta_{\kappa_{n}}\left(  k\right)  \mathrm{d}%
k,\ \ \ N<\infty,
\]
where the partition $\left(  I_{n}\right)  $ is chosen the same way as in
Section \ref{discrete oper}. Proposition \ref{finite neg spec}, part 1, then
applies with $S_{q}=\left\{  R,0\right\}  $, d$\sigma=$d$\rho_{N}\geq0$, and
$\mathbf{K}$ $=\mathbf{K}_{N}$, where the entries are given by%
\[
\left(  \mathbf{K}_{N}\right)  _{mn}=\rho\left(  I_{n}\right)  e^{8\kappa
_{n}^{3}t}\int_{x}^{\infty}\psi\left(  s,t,\mathrm{i}\kappa_{m}\right)
\psi\left(  s,t,\mathrm{i}\kappa_{n}\right)  \mathrm{d}s.
\]
Then%
\[
q_{\rho_{N}}\left(  x,t\right)  =q\left(  x,t\right)  -2\partial_{x}^{2}%
\log\det\left\{  I+\mathbf{K}_{N}\left(  x,t\right)  \right\}
\]
is a step-type KdV solution. By Proposition \ref{prop on log-det} we have%
\begin{align}
q_{\rho_{N}}\left(  x,t\right)   &  =q\left(  x,t\right) \label{q int}\\
&  +2\left[  \int\psi_{\rho_{N}}\left(  x,t;\mathrm{i}s\right)  \psi\left(
x,t;\mathrm{i}s\right)  \mathrm{d}\rho_{N}\left(  s\right)  \right]
^{2}\nonumber\\
&  +4\int\psi_{\rho_{N}}\left(  x,t;\mathrm{i}s\right)  \psi^{\prime}\left(
x,t;\mathrm{i}s\right)  \mathrm{d}\rho_{N}\left(  s\right)  ,\nonumber
\end{align}
where%
\begin{equation}
\psi_{\rho_{N}}\left(  x,t,k\right)  =\psi\left(  x,t,k\right)  -%
{\displaystyle\sum\limits_{n=1}^{N}}
\rho\left(  I_{n}\right)  e^{8\kappa_{n}^{3}t}y_{n}\left(  x,t\right)
\int_{x}^{\infty}\psi\left(  s,t,k\right)  \psi\left(  s,t,\mathrm{i}%
\kappa_{n}\right)  \mathrm{d}s, \label{psi int}%
\end{equation}
and%
\[
\left(  y_{n}\right)  =\left(  I+\mathbf{K}_{N}\right)  ^{-1}\boldsymbol{\psi
},\ \ \boldsymbol{\psi}:=\left(  \psi\left(  \cdot,\mathrm{i}\kappa
_{n}\right)  \right)  .
\]
Note that the integrals in (\ref{q int}) are actually finite sums.

By Proposition \ref{Prop on limit} with $g=\psi,$ $S=\left(  x,\infty\right)
$ we can pass in (\ref{q int}) and (\ref{psi int}) to the limit as
$N\rightarrow\infty$. Thus we have both, the weak convergence of $\left\{
R,\mathrm{d}\rho_{N}\right\}  $ to $\left\{  R,\mathrm{d}\rho\right\}  $ and
point-wise convergence of $q_{\rho_{N}}\left(  x,t\right)  $ to
\begin{align}
q_{\rho}\left(  x,t\right)   &  =q\left(  x,t\right) \label{q_r}\\
&  +2\left[  \int\psi_{\rho}\left(  x,t;\mathrm{i}s\right)  \psi\left(
x,t;\mathrm{i}s\right)  \mathrm{d}\rho_{t}\left(  s\right)  \right]
^{2}+4\int\psi_{\rho}\left(  x,t;\mathrm{i}s\right)  \psi^{\prime}\left(
x,t;\mathrm{i}s\right)  \mathrm{d}\rho_{t}\left(  s\right)  .\nonumber
\end{align}
By Definition \ref{step-type}, $q_{\rho}\left(  x,t\right)  $ is a step-type
KdV solution, which proves (\ref{DT for psi}) and (\ref{DT for q}) for
$S_{q}=\left\{  R,0\right\}  $ and d$\sigma=$d$\rho\geq0$ if we show that
$q_{\rho}\left(  x,t\right)  $ is short-range at $+\infty$ for each $t\geq0$.
This will be done later. The solubility of the Fredholm integral equation
(\ref{Fredholm}), which we rewrite in the form%
\begin{equation}
y+\mathbb{K}y=\psi, \label{int eq}%
\end{equation}
where $\mathbb{K}$ is the integral operator on $L^{2}\left(  \mathrm{d}%
\rho\right)  $ with the kernel%
\[
K\left(  \alpha,s;x,t\right)  :=\int_{x}^{\infty}\psi\left(  z,t;\mathrm{i}%
\alpha\right)  \psi\left(  z,t;\mathrm{i}s\right)  \mathrm{d}z,
\]
is obvious as $\mathbb{K}$ is clearly positive (see Section
\ref{discrete oper}) and hence $I+\mathbb{K}$ is positive definite. Since the
Hilbert-Schmidt norm $\left\vert \left\vert \mathbb{K}\left(  x,t\right)
\right\vert \right\vert _{2}$ and the $L^{2}$ norm $\left\vert \left\vert
\psi\left(  x,t;\cdot\right)  \right\vert \right\vert $ are small for large
$x$ and any fixed $t\geq0$, we immediately see from (\ref{int eq}) that
$\left\vert \left\vert y\left(  \cdot,x,t\right)  \right\vert \right\vert $ is
also small as $x\rightarrow\infty$ and (\ref{DT for psi}) readily yields%
\begin{equation}
\psi_{\rho}\left(  x,t;k\right)  =\psi\left(  x,t;k\right)  \left[  1+o\left(
1\right)  \right]  \rightarrow0,\ \ \ x\rightarrow\infty
,\ \ \ \operatorname{Im}k\geq0. \label{decay of y}%
\end{equation}
If we now show that $\psi_{\rho}$ solves the the Schr\"{o}dinger equation, it
will be its right Jost solution. We rely on the following general fact
(directly verifiable): if the Wronskian identity (\ref{W}) holds for two
functions $f_{\lambda}\left(  x\right)  $, $f_{\nu}\left(  x\right)  $ then
$f_{\lambda}^{\prime\prime}\left(  x\right)  /f_{\lambda}\left(  x\right)
+\lambda^{2}$ is independent of $\lambda$ and hence is equal to some $q\left(
x\right)  .$ Thus, $f_{\lambda}$ solves the Schr\"{o}dinger equation
$-f^{\prime\prime}+q\left(  x\right)  f=\lambda^{2}f$.

By Proposition \ref{finite neg spec} $\psi_{\rho_{N}}$ is a right Jost
solution and hence is subject to (\ref{W}). But as we have shown, $\psi
_{\rho_{N}}$ converges in $L^{2}\left(  \mathrm{d}\rho\right)  $ to some
$\psi_{\rho}$ and hence there is a subsequence convergent almost everywhere.
It follows that we can pass in (\ref{W}) to the limit and therefore
$\psi_{\rho}$ is a solution to the Schr\"{o}dinger equation\footnote{Note it
can also be shown by a direct but rather involved inspection.}. By Theorem
\ref{Darboux for RHP}, (\ref{y}), we can claim that $\psi_{\rho}\left(
x,t;\mathrm{i}s\right)  $ $=y\left(  s,x,t\right)  $, $s\in
\operatorname*{Supp}\rho$, which also proves property (1). We are ready to
show now that%
\begin{align*}
Q\left(  x\right)   &  :=q_{\rho}\left(  x,t\right)  -q\left(  x,t\right)  \\
&  =2\left[  \int\psi_{\rho}\left(  x,t;\mathrm{i}s\right)  \psi\left(
x,t;\mathrm{i}s\right)  \mathrm{d}\rho_{t}\left(  s\right)  \right]
^{2}+4\int\psi_{\rho}\left(  x,t;\mathrm{i}s\right)  \psi^{\prime}\left(
x,t;\mathrm{i}s\right)  \mathrm{d}\rho_{t}\left(  s\right)
\end{align*}
is in $L_{1}^{1}\left(  +\infty\right)  $. This will conclude the proof that
$q_{\rho}$ is indeed a step-type KdV solution. But due to (\ref{decay of y})
$Q\in L_{1}^{1}\left(  +\infty\right)  $ if%
\[
Q_{0}\left(  x\right)  :=2\left[  \int\psi\left(  x,t;\mathrm{i}s\right)
\psi\left(  x,t;\mathrm{i}s\right)  \mathrm{d}\rho_{t}\left(  s\right)
\right]  ^{2}+4\int\psi\left(  x,t;\mathrm{i}s\right)  \psi^{\prime}\left(
x,t;\mathrm{i}s\right)  \mathrm{d}\rho_{t}\left(  s\right)
\]
is in $L_{1}^{1}\left(  +\infty\right)  $. The latter immediately follows from
Lemma \ref{lemma on Jost} as each term of $Q_{0}$ is subject to its conditions.

It remains to prove the properties. (1) has already been proven. (2) clearly
holds as each term in $Q\left(  x\right)  $ is at least absolutely continuous.
The same applies to the derivatives if $q$ is differentiable sufficient number
of times. (3) is also obvious since both $\psi_{\rho}\left(  x,t;\mathrm{i}%
s\right)  $, $\psi\left(  x,t;\mathrm{i}s\right)  $ decay exponentially as
$x\rightarrow\infty$ for every $s>0$. Thus if the sets $\operatorname*{Supp}%
\rho$ and $\left\{  0\right\}  $ are separated then each term in $Q\left(
x\right)  $ decay exponentially. To show (5) we recall the following general
representation of the diagonal Green's function%
\begin{equation}
G\left(  k^{2},x\right)  =\frac{f_{+}\left(  x,k\right)  f_{-}\left(
x,k\right)  }{W\left\{  f_{+}\left(  x,k\right)  ,f_{-}\left(  x,k\right)
\right\}  },\label{G}%
\end{equation}
where $f_{\pm}$ are Weyl solutions at $\pm\infty$ respectively. Take $f_{\pm}$
as in Theorem \ref{Darboux for RHP}: $f_{+}=\psi$, $f_{-}=\varphi$. Then
(\ref{G}) reads%
\begin{equation}
G\left(  k^{2},x\right)  =-\frac{\varphi\left(  x,k\right)  \psi\left(
x,k\right)  }{2\mathrm{i}k}.\label{G1}%
\end{equation}
As is well-known, $G\left(  \lambda,x\right)  $ is a Herglotz function for
each $x$, that is an analytic function from $\operatorname{Im}\lambda>0$ into
itself. Such functions may have singularities only on $\operatorname{Im}%
\lambda=0$ of at most simple pole type. In the context of Schr\"{o}dinger
operators, a pole type singularity may only occur at a bound state. By Theorem
\ref{Darboux for RHP} all pole type singularities on the right hand side of
(\ref{G1}) come from singularities of $\varphi$, which, in turn, coincide with
pure points of $\rho$. This proves (4) and Case 1 is proven now.

Case 2 (removing negative spectrum). Suppose that $\left\{  R,\mathrm{d}%
\rho\right\}  $ are scattering data for a step-type KdV solution and
d$\sigma=\left.  \mathrm{d}\rho\right\vert _{\Delta}$, where $\Delta
\subseteq\operatorname*{Supp}\rho$. Let d$\sigma_{N}$ be a discretization of
d$\sigma$ in Case 1. By Proposition \ref{finite neg spec} we construct a
step-type KdV solution $q_{-\sigma_{N}}$ by (\ref{q-}) and the associated Jost
solution $\psi_{-\sigma_{N}}$ by (\ref{ksi-}). As in Case 1, we can pass to
the limit in Proposition \ref{finite neg spec} as d$\sigma_{N}\rightarrow
$d$\sigma$, the equations%
\[
\left(  \psi_{\sigma}\right)  _{-\sigma}=\psi,\text{ \ \ }\left(  q_{\sigma
}\right)  _{-\sigma}=q
\]
being preserved in the limit. This completes our proof of Theorem
\ref{MainThm}.

We are now in the position to complete the proof of Theorem
\ref{Darboux for RHP}. It is just enough to notice that $\widetilde{\psi}$ and
$\widetilde{q}$ in Theorem \ref{Darboux for RHP} appear as $\psi_{\sigma}$ and
$q_{\sigma}$ in Theorem \ref{MainThm}.

\begin{remark}
One of the referees asked if the process of adding negative spectrum is
commutative. I.e. adding $\sigma_{1}$ and then $\sigma_{2}$ produces the same
result as adding $\sigma_{2}$ and then $\sigma_{1}$. We are positive that this
is indeed the case. Take in (\ref{q}) two one point measures $\mathrm{d}%
\sigma_{n}\left(  s\right)  =c_{n}^{2}\delta\left(  s-\kappa_{n}\right)
\mathrm{d}s$, $\kappa_{n}>0$, $n=1,2$. We have two matrices $\mathbf{K}_{12}$
and $\mathbf{K}_{12}$%
\begin{align*}
\mathbf{K}_{12}  &  =\int_{x}^{\infty}%
\begin{pmatrix}
c_{1}^{2}\psi\left(  s,\mathrm{i}\kappa_{1}\right)  ^{2} & c_{1}c_{2}%
\psi\left(  s,\mathrm{i}\kappa_{1}\right)  \psi\left(  s,\mathrm{i}\kappa
_{2}\right) \\
c_{1}c_{2}\psi\left(  s,\mathrm{i}\kappa_{1}\right)  \psi\left(
s,\mathrm{i}\kappa_{2}\right)  & c_{2}^{2}\psi\left(  s,\mathrm{i}\kappa
_{2}\right)  ^{2}%
\end{pmatrix}
\mathrm{d}s,\\
&  \text{(adding }\sigma_{1}\text{ and then }\sigma_{2}\text{)}%
\end{align*}%
\begin{align*}
\mathbf{K}_{21}  &  =\int_{x}^{\infty}%
\begin{pmatrix}
c_{2}^{2}\psi\left(  s,\mathrm{i}\kappa_{2}\right)  ^{2} & c_{1}c_{2}%
\psi\left(  s,\mathrm{i}\kappa_{1}\right)  \psi\left(  s,\mathrm{i}\kappa
_{2}\right) \\
c_{1}c_{2}\psi\left(  s,\mathrm{i}\kappa_{1}\right)  \psi\left(
s,t,\mathrm{i}\kappa_{2}\right)  & c_{1}^{2}\psi\left(  s,t,\mathrm{i}%
\kappa_{1}\right)  ^{2}%
\end{pmatrix}
\mathrm{d}s.\\
&  \text{(adding }\sigma_{2}\text{ and then }\sigma_{1}\text{)}%
\end{align*}
Since%
\begin{align*}
&
\begin{pmatrix}
c_{2}^{2}\psi\left(  s,\mathrm{i}\kappa_{2}\right)  ^{2} & c_{1}c_{2}%
\psi\left(  s,\mathrm{i}\kappa_{1}\right)  \psi\left(  s,\mathrm{i}\kappa
_{2}\right) \\
c_{1}c_{2}\psi\left(  s,\mathrm{i}\kappa_{1}\right)  \psi\left(
s,t,\mathrm{i}\kappa_{2}\right)  & c_{1}^{2}\psi\left(  s,t,\mathrm{i}%
\kappa_{1}\right)  ^{2}%
\end{pmatrix}
\\
&  =%
\begin{pmatrix}
0 & 1\\
1 & 0
\end{pmatrix}%
\begin{pmatrix}
c_{1}^{2}\psi\left(  s,\mathrm{i}\kappa_{1}\right)  ^{2} & c_{1}c_{2}%
\psi\left(  s,\mathrm{i}\kappa_{1}\right)  \psi\left(  s,\mathrm{i}\kappa
_{2}\right) \\
c_{1}c_{2}\psi\left(  s,\mathrm{i}\kappa_{1}\right)  \psi\left(
s,t,\mathrm{i}\kappa_{2}\right)  & c_{2}^{2}\psi\left(  s,\mathrm{i}\kappa
_{2}\right)  ^{2}%
\end{pmatrix}%
\begin{pmatrix}
0 & 1\\
1 & 0
\end{pmatrix}
,
\end{align*}
we conclude that%
\[
\mathbf{K}_{21}=%
\begin{pmatrix}
0 & 1\\
1 & 0
\end{pmatrix}
\mathbf{K}_{12}%
\begin{pmatrix}
0 & 1\\
1 & 0
\end{pmatrix}
.
\]
Hence $\det\left(  I+\mathbf{K}_{21}\right)  =\det\left(  I+\mathbf{K}%
_{12}\right)  $ and therefore%
\begin{align*}
q_{\sigma_{1}\sigma_{2}}\left(  x,t\right)   &  =q\left(  x,t\right)
-2\partial_{x}^{2}\log\det\left\{  I+\mathbf{K}_{12}\left(  x,t\right)
\right\} \\
&  =q\left(  x,t\right)  -2\partial_{x}^{2}\log\det\left\{  I+\mathbf{K}%
_{21}\left(  x,t\right)  \right\}  =q_{\sigma_{2}\sigma_{1}}\left(
x,t\right)  .
\end{align*}
This simple computation can be easily extended to two general discrete
measures $\sigma_{1},\sigma_{2}$ and then we apply our density argument to go
over to arbitrary measures.
\end{remark}

\section{Examples \label{sect on reflectionless}}

In this section we offer two examples. The first one is a new derivation of
the KdV solution and the other one is an explicit construction of a step-type
potential (KdV solution) which has the same norming measure $\rho$ as the pure
step potential (\ref{pure step}) but zero reflection coefficient.

\subsection{Classical pure soliton solution}

We show that Theorem \ref{MainThm} immediately recovers the well-known
classical formula for pure $N$ soliton solution \cite{MarchBook2011}%
\begin{equation}
q\left(  x,t\right)  =-2\partial_{x}^{2}\log\det\left(  \delta_{mn}+c_{n}%
^{2}\dfrac{\mathrm{e}^{-\left(  \kappa_{m}+\kappa_{n}\right)  x+8\left(
\kappa_{m}^{3}+\kappa_{n}^{3}\right)  t}}{\kappa_{m}+\kappa_{n}}\right)  .
\label{classica pure solition}%
\end{equation}
Indeed, take in Theorem \ref{MainThm} $S_{q}=0$ (the zero background) and a
discrete measure $\rho$ given by%
\[
\mathrm{d}\rho\left(  k,t\right)  =%
{\displaystyle\sum\limits_{n=1}^{N}}
c_{n}\left(  t\right)  ^{2}\delta_{\kappa_{n}}\left(  k\right)  \mathrm{d}%
k,\ \ \ c_{n}\left(  t\right)  =c_{n}\mathrm{e}^{4\kappa_{n}^{3}t}%
\]
In this case $\psi\left(  x,k\right)  =\exp\left(  \mathrm{i}kx\right)  $ and
the Fredholm integral equation (\ref{Fredholm}) becomes a linear system,%
\[
y_{m}+%
{\displaystyle\sum\limits_{n=1}^{N}}
K_{mn}\left(  x,t\right)  y_{n}=\mathrm{e}^{-\kappa_{m}x},
\]
where%
\[
K_{mn}\left(  x,t\right)  =c_{n}^{2}\mathrm{e}^{8\kappa_{n}^{3}t}%
\frac{e^{-\left(  \kappa_{m}+\kappa_{n}\right)  x}}{\kappa_{m}+\kappa_{n}}.
\]
Furthermore,%
\begin{align}
\psi_{\rho}\left(  x,t;\kappa_{m}\right)   &  =\mathrm{e}^{-\kappa_{m}x}-%
{\displaystyle\sum\limits_{n=1}^{N}}
c_{n}^{2}\mathrm{e}^{8\kappa_{n}^{3}t}\frac{\mathrm{e}^{-\left(  \kappa
_{m}+\kappa_{n}\right)  x}}{\kappa_{m}+\kappa_{n}}y_{n}\left(  x,t\right)
\label{n=1}\\
&  =\mathrm{e}^{-\kappa_{m}x}\left[  1-%
{\displaystyle\sum\limits_{n=1}^{N}}
c_{n}^{2}\mathrm{e}^{8\kappa_{n}^{3}t-\kappa_{n}x}\frac{y_{n}\left(
x,t\right)  }{\kappa_{m}+\kappa_{n}}\right] \nonumber
\end{align}
and thus for the KdV\ solution we have%
\begin{align}
q_{\rho}\left(  x,t\right)   &  =2\left[
{\displaystyle\sum\limits_{n=1}^{N}}
c_{n}^{2}\mathrm{e}^{8\kappa_{n}^{3}t-\kappa_{n}x}\psi_{\rho}\left(
x,t;\kappa_{n}\right)  \right]  ^{2}\label{pure soliton}\\
&  -4%
{\displaystyle\sum\limits_{n=1}^{N}}
\kappa_{n}c_{n}^{2}\mathrm{e}^{8\kappa_{n}^{3}t-\kappa_{n}x}\psi_{\rho}\left(
x,t;\kappa_{n}\right)  .\nonumber
\end{align}
The formula (\ref{pure soliton}) is a new derivation of
(\ref{classica pure solition}). Due to (\ref{q}), (\ref{pure soliton}) is
equivalent to (\ref{classica pure solition}). We note that the representation
(\ref{classica pure solition}) is not as convenient for asymptotic analysis
for large $x$ and $t$ as the methods based on the Riemann-Hilbert problem (see
e.g. \cite{GT09}). We are hopeful that for similar reasons (\ref{pure soliton}%
) could be a suitable starting point to do asymptotic analysis that could work
for arbitrary enough measures.

We conclude this section with a simple computation showing how (\ref{n=1})
implies the famous one soliton solution. Taking in (\ref{n=1}) $\rho\left(
k\right)  =c^{2}\delta\left(  k-\mathrm{i}\kappa\right)  $
\[
\psi_{\rho}\left(  x,t;\mathrm{i}\kappa\right)  =e^{-\kappa x}\frac{2\kappa
}{2\kappa+c^{2}\exp\left(  8\kappa^{3}t-2\kappa x\right)  }%
\]
and substituting it in (\ref{pure soliton}) implies%
\begin{align*}
q_{\rho}\left(  x,t\right)   &  =2\left[  \frac{2c^{2}\exp\left(  8\kappa
^{3}t-2\kappa x\right)  }{2\kappa+2c^{2}\exp\left(  8\kappa^{3}t-2\kappa
x\right)  }\right]  ^{2}-4\kappa\frac{2c^{2}\exp\left(  8\kappa^{3}t-2\kappa
x\right)  }{2\kappa+2c^{2}\exp\left(  8\kappa^{3}t-2\kappa x\right)  }\\
&  =-2\kappa^{2}\operatorname{sech}^{2}\left(  4\kappa^{3}t-\kappa x+\log
\frac{c}{\sqrt{2\kappa}}\right)  ,
\end{align*}
which is the one soliton solution, as expected.

\subsection{Reflectionless step-type potential}

As was discussed in section \ref{main results}, such potentials naturally
appear in the study of soliton gases. A pure step potential (\ref{pure step})
serves as a model of soliton condensate but apparently it is not
reflectionless. In this subsection we offer a construction that produces a
step-type potential that has the same norming measure as a pure step potential
but reflectionless. To this end, take in Theorem \ref{MainThm} $S_{q}=0$ (the
zero background) and the measure $\rho$ with density d$\rho\left(  k\right)
/$d$k=2k\sqrt{1-k^{2}}$ supported on $\left[  0,1\right]  $. Apparently,%
\[
\int_{0}^{1}\mathrm{d}\rho\left(  k\right)  /k<\infty,\ \ \ \mathrm{d}\rho
\geq0,
\]
and Theorem \ref{MainThm} applies. In this case $\psi\left(  x,k\right)
=\exp\left(  \mathrm{i}kx\right)  $ and%
\[
K\left(  k,\mathrm{i}s,x\right)  =\frac{\mathrm{e}^{\left(  \mathrm{i}%
k-s\right)  x}}{s-\mathrm{i}k},
\]
are independent of $t$. Then the Fredholm integral equation (\ref{Fredholm})
turns into%
\begin{equation}
Y\left(  \alpha\right)  +\int_{0}^{1}2s\sqrt{1-s^{2}}\frac{\mathrm{e}%
^{8s^{3}t-2sx}}{s+\alpha}Y\left(  s\right)  \mathrm{d}s=1,\ \ \ \alpha
\in\left[  0,1\right]  , \label{Y}%
\end{equation}
where $Y\left(  s\right)  =y\left(  s\right)  $e$^{sx}$. Eq (\ref{Y}) has a
unique solution $Y\left(  \alpha;x,t\right)  $ and for the potential we have%
\begin{align}
q_{\rho}\left(  x,t\right)   &  =8\left[  \int_{0}^{1}s\sqrt{1-s^{2}}%
e^{-2sx}Y\left(  s;x,t\right)  \mathrm{d}s\right]  ^{2}\label{q ref}\\
&  -8\int_{0}^{1}s^{2}\sqrt{1-s^{2}}e^{-2sx}Y\left(  s;x,t\right)
\mathrm{d}s,\nonumber
\end{align}
which is a reflectionless step-type KdV solution with the scattering data
\[
S_{q_{\rho}}=\left\{  0,\mathrm{d}\rho\right\}  .
\]
By Theorem \ref{MainThm} $q_{\rho}\left(  x,t\right)  $ decays at $+\infty$
sufficiently fast but it cannot be readily seen how it behaves at $-\infty$.
In fact $q_{\rho}\rightarrow-1$ as $x\rightarrow-\infty$. To show this one
needs to perform a standard transformation of the QARHP in Theorem
\ref{Darboux for RHP} (see e.g. \cite{Grava21}). We will come back to this elsewhere.

\section{Appendix}

\subsection{Proof of Proposition \ref{Prop on limit}}

\begin{lemma}
Let $a_{1}\left(  x,t\right)  $ and $a_{2}\left(  x,t\right)  $ be real
continuos functions and%
\[
A\left(  x,y\right)  =\int_{S}a\left(  x,s\right)  a\left(  y,s\right)
\mathrm{d}s,B\left(  x,y\right)  =\int_{S}b\left(  x,s\right)  b\left(
y,s\right)  \mathrm{d}s.
\]
Then for the Hilbert-Schmidt norm of the integral operators%
\[
\left(  \mathbb{A}f\right)  \left(  x\right)  =\int A\left(  x,y\right)
f\left(  y\right)  \mathrm{d}\mu\left(  y\right)
\]
we have%
\begin{align*}
\left\vert \left\vert \mathbb{A-B}\right\vert \right\vert _{2}  &  =\left\vert
\left\vert \left\vert \mathbb{A}-\mathbb{B}\right\vert \right\vert \right\vert
\\
&  \leq\left(  \left\vert \left\vert \left\vert a\right\vert \right\vert
\right\vert +\left\vert \left\vert \left\vert b\right\vert \right\vert
\right\vert \right)  \left\vert \left\vert \left\vert a-b\right\vert
\right\vert \right\vert .
\end{align*}

\end{lemma}

\begin{proof}
Rewrite%
\begin{align*}
&  \int_{S}\left[  a\left(  x,s\right)  a\left(  y,s\right)  -b\left(
x,s\right)  b\left(  y,s\right)  \right]  \mathrm{d}s\\
&  =\int_{S}\left[  a\left(  x,s\right)  -b\left(  x,s\right)  \right]
a\left(  y,s\right)  \mathrm{d}s+\int_{S}\left[  a\left(  y,s\right)
-b\left(  y,s\right)  \right]  b\left(  x,s\right)  \mathrm{d}s\\
&  =A_{1}\left(  x,y\right)  +A_{2}\left(  x,y\right)  .
\end{align*}
Consider the $L^{2}\left(  \mathrm{d}\mu\times\mathrm{d}\mu\right)  $ norm of
each $A_{1},A_{2}$:%
\begin{align*}
&  \int A_{1}\left(  x,y\right)  ^{2}\mathrm{d}\mu\left(  x\right)
\mathrm{d}\mu\left(  y\right) \\
&  \leq\int\int\left\{  \int_{S}\left\vert a\left(  x,s\right)  -b\left(
x,s\right)  \right\vert a\left(  y,s\right)  \mathrm{d}s\right\}
^{2}\mathrm{d}\mu\left(  x\right)  \mathrm{d}\mu\left(  y\right) \\
&  \leq\int\int\left[  \int_{S}\left\vert a\left(  x,s\right)  -b\left(
x,s\right)  \right\vert ^{2}\mathrm{d}s\right]  \left\{  \int_{S}a\left(
y,s\right)  ^{2}\mathrm{d}s\right\}  \mathrm{d}\mu\left(  x\right)
\mathrm{d}\mu\left(  y\right) \\
&  =\int\left[  \int_{S}\left\vert a\left(  x,s\right)  -b\left(  x,s\right)
\right\vert ^{2}\mathrm{d}s\right]  \mathrm{d}\mu\left(  x\right)
\int\left\{  \int_{S}a\left(  y,s\right)  ^{2}\mathrm{d}s\right\}
\mathrm{d}\mu\left(  y\right) \\
&  =\left\vert \left\vert \left\vert a-b\right\vert \right\vert \right\vert
\cdot\left\vert \left\vert \left\vert a\right\vert \right\vert \right\vert .
\end{align*}
Similarly,%
\[
\int A_{2}\left(  x,y\right)  ^{2}\mathrm{d}\mu\left(  x\right)  \mathrm{d}%
\mu\left(  y\right)  \leq\left\vert \left\vert \left\vert a-b\right\vert
\right\vert \right\vert \cdot\left\vert \left\vert \left\vert b\right\vert
\right\vert \right\vert .
\]

\end{proof}

\subsection{Auxiliary estimates for Jost solutions}

Let $m\left(  x,k\right)  =$e$^{-\mathrm{i}kx}f\left(  x,k\right)  $, where
$f\left(  x,k\right)  $ is the right Jost solution. Then \cite[page
130]{Deift79} $m\left(  x,k\right)  $ is uniformly bounded in
$\operatorname{Im}k\geq0$ for every real $x$ and for $x\geq0$
\begin{equation}
\left\vert m\left(  x,k\right)  -1\right\vert \leq\operatorname*{const}%
\cdot\frac{\int_{x}^{\infty}\left(  1+\left\vert s\right\vert \right)
\left\vert q\left(  s\right)  \right\vert \mathrm{d}s}{1+\left\vert
k\right\vert }, \label{DT1}%
\end{equation}%
\begin{equation}
\left\vert m^{\prime}\left(  x,k\right)  \right\vert \leq\operatorname*{const}%
\cdot\frac{\int_{x}^{\infty}\left\vert q\left(  s\right)  \right\vert
\mathrm{d}s}{1+\left\vert k\right\vert }, \label{DT2}%
\end{equation}
uniformly for $\operatorname{Im}k\geq0$.

\begin{lemma}
\label{lemma on Jost}Let $q\in L_{1}^{1}\left(  +\infty\right)  $ and
$f\left(  x,\mathrm{i}s\right)  =$e$^{-sx}m\left(  x,\mathrm{i}s\right)
,s\geq0,$ where $m$ is subject to (\ref{DT1}) and (\ref{DT2}). Suppose that a
finite measure $\sigma$ supported on a finite subset of $\mathbb{R}_{+}$
satisfies%
\[
\int\mathrm{d}\sigma\left(  s\right)  /s<\infty.
\]
Then%
\[
F\left(  x\right)  :=\int f\left(  x,\mathrm{i}s\right)  ^{2}\mathrm{d}%
\sigma\left(  s\right)  \in L^{1}\left(  +\infty\right)
\]
and $F\left(  x\right)  ^{2}$, $F^{\prime}\left(  x\right)  $ are both in
$L_{1}^{1}\left(  +\infty\right)  $.
\end{lemma}

\begin{proof}
Without loss of generality we may assume that d$\sigma\geq0$. Observe that
(\ref{DT1}) and (\ref{DT2}) imply $m\left(  x,k\right)  =1+o\left(  1\right)
$, $m^{\prime}\left(  x,k\right)  =o\left(  1/x\right)  $ as $x\rightarrow
\infty$ uniformly in $\operatorname{Im}k\geq0$. Therefore%
\begin{equation}
F\left(  x\right)  =\int e^{-2sx}m\left(  x,\mathrm{i}s\right)  ^{2}%
\mathrm{d}\sigma\left(  s\right)  =F_{0}\left(  x\right)  \cdot\left(
1+o\left(  1\right)  \right)  ,\ \ \ x\rightarrow\infty, \label{F}%
\end{equation}%
\begin{align}
F^{\prime}\left(  x\right)   &  =-2\int\mathrm{e}^{-2sx}m\left(
x,\mathrm{i}s\right)  ^{2}s\mathrm{d}\sigma\left(  s\right)  +2\int%
\mathrm{e}^{-2sx}m\left(  x,\mathrm{i}s\right)  m^{\prime}\left(
x,\mathrm{i}s\right)  \mathrm{d}\sigma\left(  s\right) \label{F'}\\
&  =F_{0}^{\prime}\left(  x\right)  \cdot\left(  1+o\left(  1\right)  \right)
+F_{0}\left(  x\right)  \cdot o\left(  1/x\right)  ,\ \ \ x\rightarrow
\infty,\nonumber
\end{align}
where%
\[
F_{0}\left(  x\right)  :=\int\mathrm{e}^{-2sx}\mathrm{d}\sigma\left(
s\right)  .
\]

We show first that $F_{0}\in L^{1}\left(  +\infty\right)  $. For any finite
$a$ we have%
\begin{align*}
\int_{a}^{\infty}F_{0}\left(  x\right)  \mathrm{d}x  &  =\int\left[  \int%
_{a}^{\infty}\mathrm{e}^{-2sx}\mathrm{d}x\right]  \mathrm{d}\sigma\left(
s\right) \\
&  =\int\mathrm{e}^{-2as}\frac{\mathrm{d}\sigma\left(  s\right)  }{2s}\leq
\max_{s\in\operatorname*{Supp}\sigma}\mathrm{e}^{-2as}\cdot\int\frac
{\mathrm{d}\sigma\left(  s\right)  }{2s}<\infty.
\end{align*}
Since $\operatorname*{Supp}\sigma\subset\mathbb{R}_{+}$ and $\sigma$ is
finite, it immediately follows that (1) $F_{0}\in L^{1}\left(  +\infty\right)
$ and (2)\ in the computations below we can set for simplicity $a=0$. Due to
(\ref{F}), we conclude that $F\in L^{1}\left(  +\infty\right)  $. Let us show
that $F^{2}\in L_{1}^{1}\left(  +\infty\right)  $. Indeed, due to (\ref{F})
again, it is enough to show that $F_{0}^{2}\in L_{1}^{1}\left(  +\infty
\right)  $:%
\[
\int_{0}^{\infty}xF_{0}^{2}\left(  x\right)  \mathrm{d}x\leq\max_{x\geq
0}\left[  xF_{0}\left(  x\right)  \right]  \int_{0}^{\infty}F_{0}\left(
x\right)  \mathrm{d}x<\infty,
\]
as $F_{0}$ is monotonic function from $L^{1}\left(  +\infty\right)  $.

We now show that%
\[
F_{0}^{\prime}\left(  x\right)  =-2\int\mathrm{e}^{-2sx}s\mathrm{d}%
\sigma\left(  s\right)  \in L_{1}^{1}\left(  +\infty\right)  .
\]
Indeed, integrating by parts, one has%
\begin{align*}
\int_{0}^{\infty}x\left\vert F_{0}^{\prime}\left(  x\right)  \right\vert
\mathrm{d}x  &  =2\int_{0}^{\infty}\left[  \int\mathrm{e}^{-2sx}%
s\mathrm{d}\sigma\left(  s\right)  \right]  x\mathrm{d}x\\
&  =2\int\left[  \int_{0}^{\infty}\mathrm{e}^{-2sx}x\mathrm{d}x\right]
s\mathrm{d}\sigma\left(  s\right) \\
&  =2\int\left(  \frac{1}{2s}\right)  ^{2}s\mathrm{d}\sigma\left(  s\right)
=\int\frac{\mathrm{d}\sigma\left(  s\right)  }{2s}<\infty.
\end{align*}
Therefore, it follows from (\ref{F'}) that $F_{0}^{\prime}\in L_{1}^{1}\left(
+\infty\right)  $.
\end{proof}

\section{Acknowledgment}

We are grateful to Oleg Safronov for some helpful insights. We are
particularly thankful to the referees for numerous comments and questions
leading to a substantial improvement of the paper.

\end{document}